\documentclass[10pt,twocolumn,twoside,english]{IEEEtran}
\usepackage[T1]{fontenc}
\usepackage{mathrsfs}
\usepackage{bm}
\usepackage{amsmath}
\usepackage{amsthm}
\usepackage{amssymb}
\usepackage{float}
\usepackage{caption}
\usepackage{graphicx}
 \PassOptionsToPackage{normalem}{ulem}

\makeatletter

\theoremstyle{plain}
\newtheorem{theorem}{\protect\theoremname}
  \theoremstyle{plain}
  
  \theoremstyle{plain}
  
  \theoremstyle{plain}
   \newtheorem{lemma}{\protect\lemmaname}
  \theoremstyle{remark}
  
   \newtheorem{assumption}{\protect\assumptionname}
\theoremstyle{assumption}
    \theoremstyle{proposition}
  \newtheorem{proposition}{\protect\propositionname}
  
\theoremstyle{algorithm}  
   \theoremstyle{plain}
   
\usepackage{amsmath}
\usepackage{amssymb}
\usepackage{graphicx,psfrag,cite,subfigure}
\usepackage[table]{xcolor}

\setkeys{Gin}{width=1.0\columnwidth}

\makeatother

\providecommand{\definitionname}{Definition}
\providecommand{\lemmaname}{Lemma}
\providecommand{\propositionname}{Proposition}
\providecommand{\remarkname}{Remark}
\providecommand{\theoremname}{Theorem}
\providecommand{\conjecturename}{Conjecture}
\providecommand{\assumptionname}{Assumption}
\providecommand{\algorithmname}{Algorithm}

\begin{document}

 \title{Design of false data injection attack on
distributed process estimation
 \thanks{Moulik Choraria is with  EPFL Switzerland.  Arpan Chattopadhyay is with the Department of Electrical Engineering and the Bharti School of Telecom Technology and Management, Indian Institute of Technology (IIT), Delhi.  Urbashi Mitra is with the Department of Electrical Engineering, University of Southern California.  Erik Strom is with the Department of Signals and Systems, Chalmers University, Sweden. Email: moulik.choraria@epfl.ch, arpanc@ee.iitd.ac.in,  ubli@usc.edu, erik.strom@chalmers.se }
 \thanks{This work was supported by the faculty seed grant and professional development allowance (PDA) of IIT Delhi.
}
\thanks{This manuscript is an extended version of our conference paper \cite{choraria2019optimal}.}
}

\author{
Moulik Choraria, Arpan Chattopadhyay,  Urbashi Mitra, Erik Strom
}

\maketitle
%
%



\ifdefined\SINGLECOLUMN
	\setkeys{Gin}{width=0.5\columnwidth}
	\newcommand{\figfontsize}{\footnotesize} 
\else
	\setkeys{Gin}{width=1.0\columnwidth}
	\newcommand{\figfontsize}{\normalsize} 
\fi

\begin{abstract} 
Herein, design of false data injection attack on a distributed cyber-physical system is considered. A stochastic process with linear dynamics and Gaussian noise is measured by multiple agent nodes, each equipped with multiple sensors. The agent nodes form a multi-hop network among themselves. Each agent node computes an estimate of the process by using its sensor observation and messages obtained from neighboring nodes, via Kalman-consensus filtering. An external attacker, capable of arbitrarily manipulating the sensor observations of some or all agent nodes, injects errors into those sensor observations. The goal of the attacker is to steer the estimates at the agent nodes as close as possible to a pre-specified value, while respecting a constraint on the attack detection probability. To this end, a constrained optimization problem is formulated to find the optimal parameter values of a certain class of linear attacks. The parameters of linear attack are learnt on-line via a combination of   stochastic approximation based update of a Lagrange multiplier, and an optimization technique  involving either the Karush-Kuhn-Tucker (KKT) conditions or   online stochastic gradient descent. The problem turns out to be convex for some special cases. Desired convergence of the proposed algorithms are proved by exploiting the convexity and properties of stochastic approximation algorithms.  Finally, numerical results demonstrate the efficacy of the  attack.
\end{abstract}

\section{Introduction}\label{section:introduction}
In recent times, there have been significant interest in designing cyber-physical systems (CPS) that  combine the cyber world and the physical world via seamless integration of sensing, computation, communication, control and learning. CPS has widespread applications such as networked  monitoring and control of industrial processes, disaster management, smart grids, intelligent transportation systems, etc. These applications critically depend on estimation of a physical process via multiple sensors over a wireless network. However, increasing use of wireless networks in sharing the sensed data has rendered the sensors vulnerable to various cyber-attacks.  In this paper, we focus on {\em false data injection} (FDI) attacks which is an integrity or deception attack where the attacker modifies the information flowing through the network \cite{mo2009secure, mo2014detecting}, in contrast  to a {\em denial-of-service} attack where the attacker blocks system resources ({\em e.g.},  wireless jamming attack  \cite{guan2018distributed}). In FDI, the attacker either breaks the cryptography of the data packets or physically manipulates the sensors ({\em e.g.}, putting a heater near a temperature sensor).

The cyber-physical systems either need to compute the process estimate in a remote estimator ({\em centralized} case), or often multiple nodes or components of the system need to estimate the same process over time via sensor observations and the information shared over a network ({\em distributed} case). 
The problem of FDI attack design and its detection has received significant  attention in recent times; attack design:  conditions for    undetectable FDI  attack \cite{chen2017optimal}, design of a linear deception attack scheme to fool the popular $\chi^2$ detector (see \cite{guo2017optimal}), optimal attack design for noiseless systems \cite{wu2018optimal}. The paper  \cite{chen2016cyber} designs an optimal attack to steer the  state of a control system to a desired target under a  constraint on the attack detection probability. On the other hand, attempts on attack detection includes centralized (and decentralized as well)     schemes for {\em noiseless} systems  \cite{pasqualetti2013attack}, coding of sensor output along with $\chi^2$ detector   \cite{miao2017coding},  comparing the sensor observations with those coming from from a few {\em known safe} sensors  \cite{li2017detection}, and the attack detection and secure estimation schemes based on innovation vectors in \cite{mishra2017secure}.  Attempts on attack-resilient state estimation include: \cite{pajic2017attack} for  {\em bounded} noise, \cite{chattopadhyay2018attack, chattopadhyay2018secure, chattopadhyay2019security} for adaptive filter design using stochastic approximation,  \cite{liu2017dynamic} that uses  sparsity models to characterize the switching location attack in a {\em noiseless} linear system and    state recovery constraints for various attack modes. FDI attack and its mitigation  in power  systems are addressed in \cite{manandhar2014detection, liang2017review, hu2017secure}.  Attack-resilient state estimation and control in noiseless systems are discussed in  \cite{nakahira2018attack} and \cite{fawzi2014secure}. Performance bound of stealthy attack in a single sensor-remote estimator system using Kalman filter was characterized in \cite{bai2017kalman}.

However, there have been very few attempts for attack  mitigation in distributed CPS, except \cite{guan2017distributed} for attack detection and secure estimation, \cite{satchidanandan2016dynamic} for attack detection in networked control system using a certain {\em dynamic watermarking} strategy, and \cite{dorfler2011distributed} for attack detection in power systems. On the other hand, the authors of \cite{moradi2019coordinated} have designed an attack scheme to maximize the network-wide estimation error, which is different from our objective of pushing the estimates across nodes towards a target value, while respecting the attack detection constraint. Also, contrary to \cite{moradi2019coordinated} which adds a simple Gaussian noise to the attacked node's observation, we focus on the class of linear attacks, and provide theoretical convergence results of our proposed  online learning based attack schemes.  To our knowledge, there has been no other attempt to theoretically design an attack strategy in distributed CPS. 
In light of these, our contributions in this paper are the following:

\begin{enumerate}
    \item Under the Kalman-consensus filter (KCF, see \cite{saber09kalman-consensus-optimality-stability}) for distributed estimation, we design a novel attack scheme that steers the estimates in all estimators towards a target value, while respecting a constraint on the attack detection probability under the popular $\chi^2$ detector adapted to the distributed setting. The attack scheme is reminiscent of the popular linear attack scheme \cite{guo2017optimal}, but the novelty lies in online learning and optimization  of the parameters in the attack algorithm via Karush-Kuhn-Tucker (KKT) conditions, multi-timescale stochastic approximation \cite{borkar08stochastic-approximation-book} and  simultaneous perturbation stochastic approximation (SPSA  \cite{spall92original-SPSA}). The attack algorithm, unlike the linear attack scheme of  \cite{guo2017optimal}, uses a non-zero mean Gaussian perturbation to modify the observation made at a node, and this non-zero mean is an affine function of the process estimate at a node. The optimization problem is cast as an online  optimization problem, where KKT conditions are used for finding the optimal attack scheme, and, alternatively, SPSA is used for online stochastic gradient descent based learning of attack parameters (see \cite[Chapter~$3$]{hazan2016introduction}). These works are also extended to the case where the attacker has access to the FDI alarm at each node.
    \item The constraint on attack detection probability is met by updating a Lagrange multiplier via stochastic approximation at a slower timescale.
    \item The dynamics of the deviation of the estimates from the target is derived analytically, which is used later to formulate the online optimization problem.
    \item Theoretical convergence results are proved for all attack design schemes proposed in this paper.
    \item Though the proposed algorithm involves on-line parameter learning, it can be used off-line to optimize the attack parameters which can then be used in real CPS.
\end{enumerate}
 
The rest of the paper is organized as follows. System model and the necessary background related to the problem are provided in Section~\ref{section:model}. Error dynamics expressions under FDI are calculated in Section~\ref{section:error-dynamics-under-attack}. Attack design algorithms  are developed in Section~\ref{section:attack-design-via-KKT} via KKT conditions, and in Section~\ref{section:attack-design-SPSA} via SPSA. Numerical results are presented in Section~\ref{section:numerical-work}, followed by the conclusions in Section~\ref{section:conclusion}. All proofs are provided in the appendices.

 \begin{figure}[t!]
 \begin{centering}
 \begin{center}
 \includegraphics[height=6cm, width=7cm]{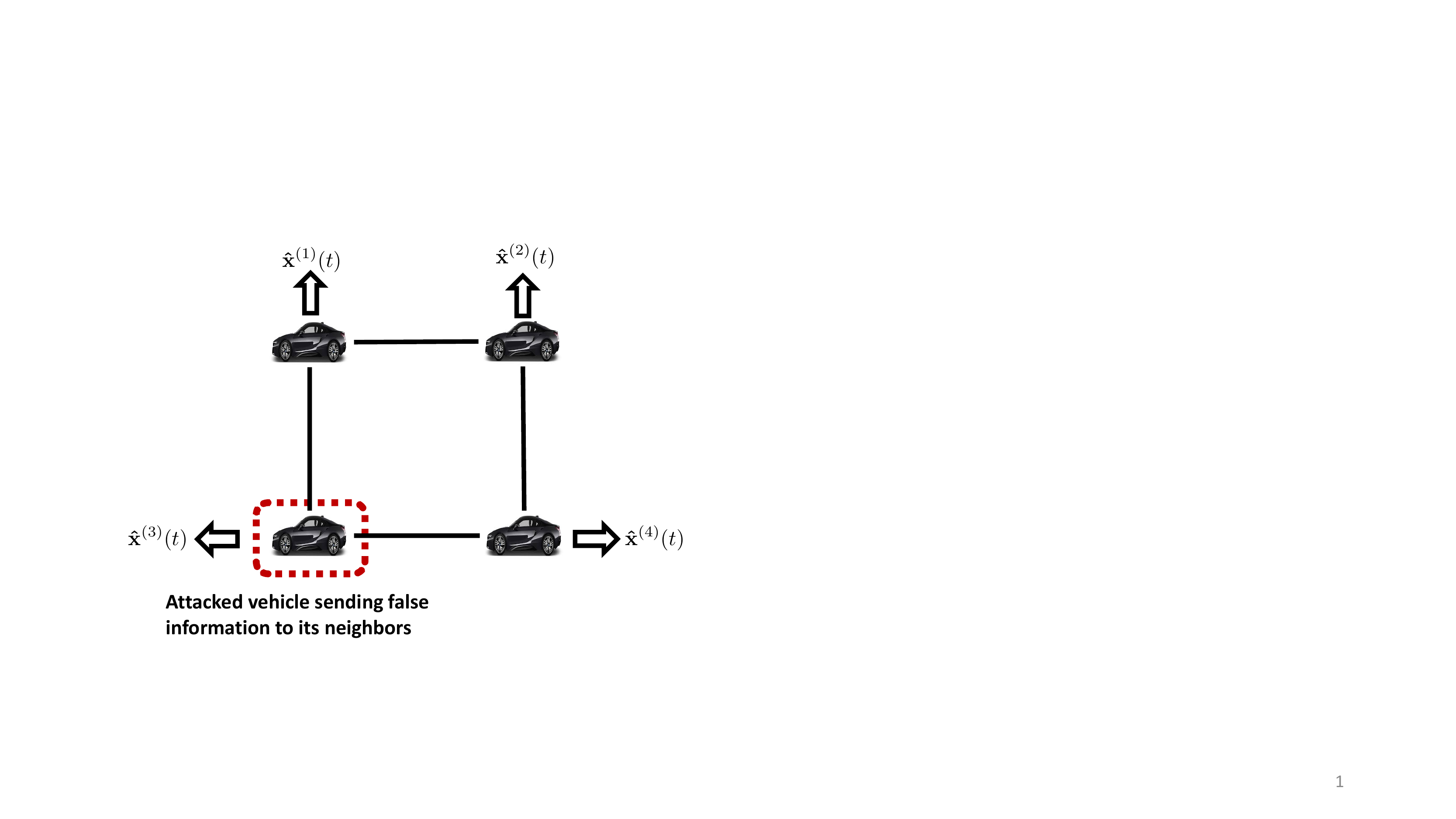}
 \end{center}
 \end{centering}
 \vspace{-5mm}
 \caption{False data injection attack in distributed  estimation.}
 \label{fig:distributed-attack}
 \vspace{-5mm}
 \end{figure}

\section{System Model}\label{section:model}
In this paper, bold capital letters, bold small letters   and capital letters with caligraphic font  will denote matrices, vectors and sets respectively.

\subsection{Sensing and  estimation model: no attack}\label{subsection:sensing-model}
We consider a connected, undirected, multi hop wireless network (see Figure~\ref{fig:distributed-attack}) of $N$ agent nodes denoted by $\mathcal{N}\doteq\{1,2,\cdots,N\}$. The set of neighboring nodes of node~$k$ is denoted by $\mathcal{N}_k$, and let $N_k \doteq |\mathcal{N}_k|$.  There is a discrete-time stochastic process $\{\bm{x}(t)\}_{t \geq 0}$ (where $\bm{x}(t) \in \mathbb{R}^{q \times 1}$ with process dimension $q$) which is a linear process with Gaussian noise evolving as follows:
\begin{equation}\label{eqn:process-equation}
\bm{x}(t+1)=\bm{A} \bm{x}(t)+\underbrace{\bm{w}(t)}_{\sim \mathcal{N}(\bm{0}, \bm{Q})} 
\end{equation}
where $\bm{w}(t)$ is  zero-mean i.i.d. Gaussian noise with covariance matrix $\bm{Q}$, and $\bm{A} \in \mathbb{R}^{q \times q}$ is the process matrix. 

Each agent node is equipped with one or more sensors which make some observation about the process. The  vector observation received at node~$k$ at time~$t$  is given by:
\begin{equation}
\bm{y}_k(t)=\bm{H}_k \bm{x}(t)+\underbrace{\bm{v}_k(t)}_{\sim \mathcal{N}(\bm{0}, \bm{R}_k)},  \label{eqn:observation-equation}
\end{equation}
where $\bm{H}_k$ is an observation matrix of appropriate dimension, and $\bm{v}_k(t)$ is a zero-mean Gaussian observation noise with covariance matrix $\bm{R}_k$, which is independent across sensors and i.i.d. across $t$. The pair $(\bm{A}, \bm{Q}^{\frac{1}{2}})$ is assumed to be stabilizable, and the pair $(\bm{A},\bm{H}_k)$ is assumed to be observable for each $1 \leq  k \leq N$.

At time $t$, each agent node~$k \in \mathcal{N}$ declares an estimate $\bm{\hat{x} }^{(k)}(t)$ using Kalman consensus filtering (KCF, see \cite{saber09kalman-consensus-optimality-stability}) which involves the following sequence of steps:
\begin{enumerate}
\item Node~$k$ computes an intermediate estimate $\bar{\bm{x}}^{(k)}(t)=\bm{A}\hat{\bm{x}}^{(k)}(t-1)$.
\item Node~$k$ broadcasts $\bar{\bm{x}}^{(k)}(t)$ to all $j \in \mathcal{N}_k$. 
\item Node~$k$ computes its final estimate of the process as:
\begin{eqnarray}\label{eqn:KCF-equation}
 \hat{\bm{x}}^{(k)}(t)&=&\bar{\bm{x}}^{(k)}(t)+\bm{G}_k (\bm{y}_k(t)-\bm{H}_k \bar{\bm{x}}^{(k)}(t)) \nonumber\\
 && +\bm{C}_k \sum_{j \in \mathcal{N}_k} (\bar{\bm{x}}^{(j)}(t)-\bar{\bm{x}}^{(k)}(t))
\end{eqnarray}
\end{enumerate}
Here $\bm{G}_k$ and $\bm{C}_k$ are the Kalman and consensus gain matrices used by node~$k$, respectively.

\subsection{The $\chi^2$ detector} \label{subsection:chi-square-detector}
Let us define the innovation vector at node~$k$ by $\bm{z}_k(t):=\bm{y}_k(t)-\bm{H}_k \bm{A} \hat{\bm{x}}^{(k)}(t-1)$.  Let us assume that, under no attack, 
$\{\bm{z}_k(t)\}_{t \geq 0}$ reaches its steady-state distribution 
 $N(\bm{0}, \bm{\Sigma}_k)$. Under a possible attack, a  standard technique (see \cite{guo2017optimal}, \cite{li2017detection}) to detect any anomaly in $\{\bm{z}_t\}_{t \geq 0}$ is the $\chi^2$ detector, which tests whether the innovation vector follows the desired Gaussian distribution. The detector {\em at each agent node} observes the innovation sequence over a pre-specified window of $J$ time-slots, and declares   an attack at time $\tau$   if  
$\sum_{t=\tau-J+1}^{\tau} \bm{z}_k(t)' \bm{\Sigma}_k^{-1} \bm{z}_k(t) \geq \eta$,  
where $\eta$ is a threshold  which can be adjusted to control the false alarm probability. The covariance matrix $\bm{\Sigma}_k$ can be computed from standard results on KCF as in \cite{saber09kalman-consensus-optimality-stability}.

\subsection{False data injection (FDI) attack}\label{subsection:FDI-attack}
At time~$t$, sensors associated to any  subset of nodes $\mathcal{A}_t \subset \mathcal{N}$ can be under attack. A node~$k \in \mathcal{A}_t$ receives an  observation:
\begin{eqnarray}\label{eqn:attack-equation}
\tilde{\bm{y}}_k(t)&=&\bm{y}_k(t)+\bm{e}_k (t)\nonumber\\
&=&\bm{H}_k \bm{x}(t)+\bm{e}_k (t)+\bm{v}_k(t),  
\end{eqnarray}
where $\bm{e}_k(t)$ is the error injected by the attacker. The attacker seeks to insert the error sequence $\{\bm{e}_k(t): k \in \mathcal{A}_t\}_{t \geq 0}$   in order to introduce error in the estimation. 
If $\mathcal{A}_t=\mathcal{A}$ for all $t$, then the attack is called a {\em static attack}, otherwise the attack is called a {\em switching location attack}. {\em We will consider only static attack in this paper, though the theory developed in this paper can be extended to switching location attack.} We assume that the attacker can observe $\hat{\bm{x}}^{(k)}(t)$ for all $1 \leq k \leq N$ once they are computed by the agent nodes. We also assume that the attacker knows the matrices $\bm{A}, \bm{Q}, \{\bm{H}_k\}_{1 \leq k \leq N}, \{\bm{R}_k\}_{1 \leq k \leq N}$.

\subsection{The optimization problem}\label{subsection:optimization-problem}
The attacker seeks to steer the estimate at each agent node as close as possible to some pre-defined value $\bm{x}^*$, while keeping the attack detection probability per unit time under some constraint value $\alpha$. The authors of  \cite{guo2017optimal}  proposed a linear injection attack to fool the $\chi^2$ detector in a centralized, remote estimation setting. Motivated by  \cite{guo2017optimal}, we also propose a linear attack, where,  at time $t$, the    sensor(s) associated with any node~$k \in \mathcal{A}$   modifies the innovation vector as $\bm{\tilde{z}_k}(t)=\bm{T}_k \bm{z}_k(t)+\bm{b}_k(t)$, where $\bm{T}_k$ is a square matrix and $\bm{b}_k(t) \sim N (\bm{\mu}_k( \bm{\theta}^{(k)}(t-1)),\bm{S}_k)$ is  independent Gaussian with its mean taken as a function of $\bm{\theta}^{(k)}(t-1) \doteq \hat{\bm{x}}^{(k)}(t-1)-\bm{x}^*$. The bias term $\bm{\mu}_k(\bm{\theta}^{(k)}(t-1))$ is assumed to take a linear form $\bm{\mu}_k(\bm{\theta}^{(k)}(t-1))=\bm{M}_k \bm{\theta}^{(k)}(t-1)+\bm{d}_k$ for suitable matrix and vector $\bm{M}_k$ and $\bm{d}_k$. This is equivalent to modifying the observation vector to $\bm{\tilde{y}_k}(t)$. If $\{\bm{T}_k, \bm{S}_k, \bm{M}_k, \bm{d}_k\}_{1 \leq k \leq N}$ is constant over time~$t$, the attack is called stationary, else non-stationary.

Note that, the probability of attack detection per unit time slot under the $\chi^2$ detector can be upper bounded as:

\footnotesize
\begin{eqnarray}\label{eqn:Markov-bound}
  P_d &=&  \limsup_{T \rightarrow \infty} \frac{1}{T+1} \sum_{\tau=0}^T  \mathbb{P} \bigg( \cup_{k=1}^N \{ \sum_{t=\tau-J+1}^{\tau} \tilde{\bm{z}}_k(t)' \bm{\Sigma}_k^{-1} \tilde{\bm{z}}_k(t) \geq \eta \} \bigg) \nonumber\\
  &\leq& \limsup_{T \rightarrow \infty} \frac{1}{T+1} \sum_{\tau=0}^T \sum_{k=1}^N \mathbb{P}(\sum_{t=\tau-J+1}^{\tau} \tilde{\bm{z}}_k(t)' \bm{\Sigma}_k^{-1} \tilde{\bm{z}}_k(t) \geq \eta) \nonumber\\
  &\leq & \limsup_{T \rightarrow \infty} \frac{1}{T+1} \sum_{\tau=0}^T \sum_{k=1}^N \frac{\mathbb{E}(\sum_{t=\tau-J+1}^{\tau} \tilde{\bm{z}}_k(t)' \bm{\Sigma}_k^{-1} \tilde{\bm{z}}_k(t))}{\eta} \nonumber\\
  &= & \frac{J}{\eta} \limsup_{T \rightarrow \infty} \frac{1}{T+1} \sum_{\tau=0}^T \sum_{k=1}^N \mathbb{E} (\tilde{\bm{z}}_k(t)' \bm{\Sigma}_k^{-1} \tilde{\bm{z}}_k(t) )
\end{eqnarray}
\normalsize 

where the two inequalities come from the union bound and the Markov inequality, respectively. Hence, the attacker seeks to solve the following constrained optimization problem:

\small
\begin{align}\label{eqn:constrained-optimization-problem} 
 && \min_{ \{\bm{T}_k, \bm{S_k}, \bm{M}_k, \bm{d}_k \}_{k=1}^N } \limsup_{T \rightarrow \infty} \frac{1}{T+1} \sum_{t=0}^T \sum_{k=1}^N \mathbb{E}||\hat{\bm{x}}^{(k)}(t)-\bm{x}^*||^2    \nonumber\\
 &\text{s.t. }& \limsup_{T \rightarrow \infty} \frac{1}{T+1} \sum_{t=0}^T \sum_{k=1}^N \mathbb{E} ( \tilde{\bm{z}}_k(t)' \bm{\Sigma}_k^{-1} \tilde{\bm{z}}_k(t) ) \leq \frac{\alpha \eta}{J} \tag{CP}
\end{align}
\normalsize

This problem can be relaxed by a Lagrange multiplier $\lambda$ to obtain the following unconstrained optimization problem:
\small
\begin{align}\label{eqn:unconstrained-optimization-problem}  
 \min_{ \{\bm{T}_k, \bm{S}_k,  \bm{M}_k, \bm{d}_k\}_{k=1}^N } && \limsup_{T \rightarrow \infty} \frac{1}{T+1} \sum_{t=0}^T \sum_{k=1}^N \mathbb{E}(||\hat{\bm{x}}^{(k)}(t)-\bm{x}^*||^2  \nonumber\\
 && +\lambda \tilde{\bm{z}}_k(t)' \bm{\Sigma}_k^{-1} \tilde{\bm{z}}_k(t) ) \tag{UP}
\end{align}
\normalsize

The following standard result tells us how to choose $\lambda$.
\begin{proposition}\label{proposition:choice-of-lambda}
 Let us consider \eqref{eqn:constrained-optimization-problem} and its relaxed version \eqref{eqn:unconstrained-optimization-problem}. If there exists a $\lambda^* \geq 0$ and matrices $\{\bm{T}_k^*, \bm{S}_k^*, \bm{M}_k^*, \bm{d}_k^*\}_{k=1}^N$ such that (i) $\{\bm{T}_k^*, \bm{S}_k^*, \bm{M}_k^*, \bm{d}_k^*\}_{k=1}^N$ is the optimal solution of \eqref{eqn:unconstrained-optimization-problem} under $\lambda=\lambda^*$, and 
 (ii) the tuple $\{\bm{T}_k^*, \bm{S}_k^*, \bm{M}_k^*, \bm{d}_k^*\}_{k=1}^N$ satisfies the constraint in \eqref{eqn:constrained-optimization-problem} with equality, then $\{\bm{T}_k^*, \bm{S}_k^*, \bm{M}_k^*, \bm{d}_k^*\}_{k=1}^N$ is an optimal solution for \eqref{eqn:constrained-optimization-problem} as well.
\end{proposition}

Proposition~\ref{proposition:choice-of-lambda} says that, if we choose an appropriate value for $\lambda^*$ and solve \eqref{eqn:unconstrained-optimization-problem},  we will obtain an optimal solution to \eqref{eqn:constrained-optimization-problem}. In this section, we provide an on-line learning algorithm to find $(\{\bm{T}_k^*, \bm{S}_k^*, \bm{M}_k^*, \bm{d}_k^*\}_{k=1}^N, \lambda^*)$. However, we will first analytically characterize the dynamics of the deviation 
$(\hat{\bm{x}}^{(k)}(t)-\bm{x}^*)$ in presence of linear attack, which will be used in developing the attack design algorithm later.

\section{Error dynamics under attack}\label{section:error-dynamics-under-attack}
Let us consider an algorithm that maintains  iterates $\{\bm{T}_k(t), \bm{U}_k(t), \bm{M}_k(t), \bm{d}_k(t)\}_{1 \leq k \leq N}$ and $\lambda(t)$ for $\{\bm{T}_k, \bm{U}_k, \bm{M}_k, \bm{d}_k\}_{1 \leq k \leq N}$ and $\lambda$, where $\bm{U}_k' \bm{U}_k \doteq \bm{S}_k$. Since it is difficult to maintain $\bm{S}_k(t)$ positive definite in an iterative algorithm, we choose to iteratively update $\bm{U}_k(t)$ and set $\bm{S}_k(t)=\bm{U}_k'(t) \bm{U}_k(t)$. 

Let us define the sigma algebra:
\begin{eqnarray}
 \mathcal{F}_{\tau} &\doteq& \sigma (\{\hat{\bm{x}}^{(k)}(t), \bm{y}_k(t), \bm{T}_k(t), \bm{U}_k(t), \bm{M}_k(t), \bm{d}_k(t),  \nonumber\\
&&\bm{b}_k(t), \lambda(t)\}_{1 \leq k \leq N}, \lambda(t): 1 \leq  t \leq \tau  )
\end{eqnarray}
This is the information  available to the attacker at time $(\tau+1)$ before   a new attack. 
However, let us assume for the sake of analysis that the attacker uses constant $\bm{T}_k, \bm{M}_k, \bm{d}_k, \bm{U}_k$ respectively, for all $k \in \{1,2,\cdots,N\}$.

 Let  $\tilde{\bm{\phi}}(t) \doteq (\hat{\bm{x}}(t)-\bm{x}(t))$, where 
$\hat{\bm{x}}(t) \doteq \mathbb{E}(\bm{x}(t)|\{ \bm{y}_k(\tau)\}_{1 \leq k \leq N, \tau \leq t})=\mathbb{E}(\bm{x}(t)|\mathcal{F}_{t})$   is the MMSE estimate of $\bm{x}(t)$ under no attack and  can be computed by the attacker using a standard Kalman filter. Clearly, $\tilde{\bm{\phi}}(t) \sim \mathcal{N}(\bm{0}, \bm{R}(t))$ where $\bm{R}(t)$ can be computed by a standard Kalman filter. Hence, given $\mathcal{F}_{t}$, $\bm{x}(t) \sim \mathcal{N}(\hat{\bm{x}}(t), \bm{R}(t))$.  Also, conditioned on $\mathcal{F}_t$, the distribution of $\bm{\phi}(t) \doteq (\bm{x}(t)-\bm{x}^*)$ is $\mathcal{N}(\hat{\bm{x}}(t)-\bm{x}^*, \bm{R}(t))$. Note that, these quantities can be computed by the attacker via a standard Kalman filter. 

Let us also recall that $\bm{\theta}^{(k)}(t) \doteq \hat{\bm{x}}^{(k)}(t)-\bm{x}^*$. 

\begin{theorem}\label{theorem:theta-and-z-evolution}
Under a constant $\{\bm{T}_k, \bm{M}_k, \bm{d}_k, \bm{U}_k\}_{1 \leq k \leq N}$, the quantity  $\mathbb{E}(||\bm{\theta}^{(k)}(t)||^2| \mathcal{F}_{t-1})$  can be expressed as  \eqref{eqn:variance-of-theta} and $\mathbb{E}(\tilde{\bm{z}}_k(t)' \bm{\Sigma}_k^{-1} \tilde{\bm{z}}_k(t)|\mathcal{F}_{t-1})$ can be expresed by   \eqref{eqn:variance-of-z}.
\end{theorem}
\begin{proof}
See Appendix~\ref{appendix:proof-of-theta-and-z-evolution}.
\end{proof}

\begin{figure*}[t!]
 
  \footnotesize
 \begin{eqnarray}\label{eqn:variance-of-theta}
  \mathbb{E}(||\bm{\theta}^{(k)}(t)||^2| \mathcal{F}_{t-1})
  &=& ||(\bm{A}-\bm{G}_k \bm{T}_k \bm{H}_k \bm{A}-N_k \bm{C}_k \bm{A})\bm{\theta}^{(k)}(t-1)+ \bm{C}_k \bm{A} \sum_{j \in \mathcal{N}_k} \bm{\theta}^{(j)}(t-1)-(\bm{I}-\bm{A})\bm{x}^*+ \bm{G}_k (\bm{M}_k \bm{\theta}^{(k)}(t-1)+\bm{d}_k)||^2 \nonumber\\
  &&+ \mbox{Tr} ( \bm{G}_k \bm{T}_k \bm{H}_k \bm{Q} \bm{H}_k' \bm{T}_k' \bm{G}_k' + \bm{G}_k \bm{S}_k \bm{G}_k' + \bm{G}_k \bm{T}_k \bm{R}_k \bm{T}_k' \bm{G}_k' ) \nonumber\\
  && + 2 \bigg((\bm{A}-\bm{G}_k \bm{T}_k \bm{H}_k \bm{A}-N_k \bm{C}_k \bm{A})\bm{\theta}^{(k)}(t-1)+ \bm{C}_k \bm{A} \sum_{j \in \mathcal{N}_k} \bm{\theta}^{(j)}(t-1)-(\bm{I}-\bm{A})\bm{x}^* + \bm{G}_k (\bm{M}_k \bm{\theta}^{(k)}(t-1)+\bm{d}_k)) \bigg)' \nonumber\\
  && \bm{G}_k \bm{T}_k \bm{H}_k \bm{A} 
  \underbrace{\mathbb{E}( \bm{\phi}(t-1) |\mathcal{F}_{t-1})}_{=\hat{\bm{x}}(t-1)-\bm{x}^*}  + \underbrace{ \mathbb{E}(|| \bm{G}_k \bm{T}_k \bm{H}_k \bm{A}\bm{\phi}(t-1)||^2 |\mathcal{F}_{t-1}) }_{=\mbox{Tr}\bigg( \bm{G}_k \bm{T}_k \bm{H}_k \bm{A} \bigg(\bm{R}(t-1)+(\hat{\bm{x}}(t-1)-\bm{x}^*)(\hat{\bm{x}}(t-1)-\bm{x}^*)' \bigg) \bm{A}'\bm{H}_k' \bm{T}_k' \bm{G}_k' \bigg)}
  \end{eqnarray}
 
  \begin{eqnarray}\label{eqn:variance-of-z}
   \mathbb{E}(\tilde{\bm{z}}_k(t)' \bm{\Sigma}_k^{-1} \tilde{\bm{z}}_k(t) | \mathcal{F}_{t-1}) 
  &=& \mbox{Tr} \bigg( \bm{\Sigma}_k^{-\frac{1}{2}} \bigg( \bm{T}_k \bm{H}_k \bm{Q} \bm{H}_k' \bm{T}_k'+ \bm{T}_k \bm{R}_k \bm{T}_k' + \bm{S}_k + \bm{T}_k \bm{H}_k \bm{A} \bm{R}(t-1) \bm{A}' \bm{H}_k' \bm{T}_k' \nonumber\\
  &&  + [\bm{T}_k \bm{H}_k \bm{A} \hat{\bm{x}}(t-1)-\bm{T}_k \bm{H}_k \bm{A} \hat{\bm{x}}^{(k)}(t-1)+\bm{M}_k \bm{\theta}^{(k)}(t-1)+\bm{d}_k ] \nonumber\\
  && [\bm{T}_k \bm{H}_k \bm{A} \hat{\bm{x}}(t-1)-\bm{T}_k \bm{H}_k \bm{A} \hat{\bm{x}}^{(k)}(t-1)+\bm{M}_k \bm{\theta}^{(k)}(t-1)+\bm{d}_k ]' \bigg) \bm{\Sigma}_k^{-\frac{1}{2}} \bigg)
 \end{eqnarray}
 \normalsize
\hrule
\end{figure*}

Note that, given $\{\bm{\theta}^{(k)}(t-1): 1 \leq k \leq N\}$, the function $\sum_{k=1}^N \mathbb{E}(||\bm{\theta}^{(k)}(t)||^2| \mathcal{F}_{t-1} )$ and $\sum_{k=1}^N \mathbb{E}(\tilde{\bm{z}}_k(t)' \bm{\Sigma}_k^{-1} \tilde{\bm{z}}_k(t)|\mathcal{F}_{t-1})$ are quadratic in $\{\bm{T}_k, \bm{U}_k, \bm{M}_k, \bm{d}_k\}_{1 \leq k \leq N}$. 
Hence,  the function 

\footnotesize
\begin{eqnarray}\label{eqn:definition-of-f}
 && f_t(\{\bm{T}_k, \bm{U}_k, \bm{M}_k, \bm{d}_k\}_{1 \leq k \leq N}, \lambda) \nonumber\\
 & \doteq & \sum_{k=1}^N \mathbb{E}(||\bm{\theta}^{(k)}(t)||^2 + \lambda \tilde{\bm{z}}_k(t)' \bm{\Sigma}_k^{-1} \tilde{\bm{z}}_k(t)| \mathcal{F}_{t-1} )
\end{eqnarray}
\normalsize

is also quadratic in $\{\bm{T}_k, \bm{U}_k, \bm{M}_k, \bm{d}_k\}_{1 \leq k \leq N}$. 
In case of  non-stationary  attack, these   results will hold w.r.t. $\{\bm{T}_k(t), \bm{U}_k(t), \bm{M}_k(t), \bm{d}_k(t)\}_{1 \leq k \leq N}$.

\begin{lemma}\label{lemma:convexity}
The function $\mathbb{E}( \tilde{\bm{z}}_k(t)' \bm{\Sigma}_k^{-1} \tilde{\bm{z}}_k(t)| \mathcal{F}_{t-1} )$ is convex in $\{ \bm{T}_k, \bm{U}_k, \bm{M}_k, \bm{d}_k\}_{1 \leq k \leq N}$. 
For fixed $\{\bm{T}_k \}_{1 \leq k \leq N}$, the functions $\mathbb{E}(||\bm{\theta}^{(k)}(t)||^2 | \mathcal{F}_{t-1} )$ and   $f_t(\{\bm{T}_k, \bm{U}_k, \bm{M}_k, \bm{d}_k\}_{1 \leq k \leq N}, \lambda)$ are convex in $\{ \bm{U}_k, \bm{M}_k, \bm{d}_k\}_{1 \leq k \leq N}$.  
\end{lemma}
\begin{proof}
See Appendix~\ref{appendix:proof-of-convexity}.
\end{proof}

\subsubsection{Stability of $\{\bm{\theta}^{(k)}(t)\}$}
Let us consider constant $\{\bm{T}_k(t), \bm{U}_k(t), \bm{M}_k(t), \bm{d}_k(t)\}_{1 \leq k \leq N}$ over time. Let us define the    matrix $\bm{M}$ consisting of $N^2$ blocks (each block is a square matrix) where:
\begin{itemize}
    \item The $(k,k)$-th block in $\bm{M}$ is $(\bm{A}-\bm{G}_k \bm{T}_k \bm{H}_k \bm{A}-N_k \bm{C}_k \bm{A})$.
    \item For $k \neq j$ and $j \in \mathcal{N}_k$, the $(k,j)$-th  block of $\bm{M}$ is $\bm{C}_k \bm{A}$.
    \item For $k \neq j$ and $j \notin \mathcal{N}_k$, the $(k,j)$-th block of $\bm{M}$ is $\bm{0}$.
\end{itemize}
\begin{lemma}\label{lemma:proof-of-stability-of-theta}
The error dynamics $\{\bm{\theta}^{(k)}(t)\}_{1 \leq k \leq N}$ is stable if the spectral radius of $\bm{M}$ is less than $1$.
\end{lemma}
\begin{proof}
 See Appendix~\ref{appendix:proof-of-stability-of-theta}.
\end{proof}
Clearly, if we choose $\bm{T}_k=\bm{I}$ for $1 \leq k \leq N$, then the $\{\bm{\theta}^{(k)}(t): 1 \leq k \leq N\}_{t \geq 0}$ process remains stable if the estimates at various nodes are stable  under no attack.
\begin{lemma}\label{lemma:stability-of-innovation}
If the spectral radius of $\bm{M}$ is less than $1$, then the $\{\tilde{z}_k(t)\}_{t \geq 0}$ process is also stable for all $1 \leq k \leq N$.
\end{lemma}
\begin{proof}
 We know that $\tilde{\bm{z}}(t)=\bm{T}_k(\bm{y}_k(t)-\bm{H}_k \bm{A} \hat{\bm{x}}^{(k)}(t-1))+\bm{b}_k(t)$. Since the true observation sequence $\{\bm{y}_k(t)\}_{t \geq 0}$ is stable, $\{\bm{b}_k(t)\}_{t \geq 0}$ is i.i.d., and $\{\hat{\bm{x}}^{(k)}(t)\}_{t \geq 0}$ is stable under FDI  (by Lemma~\ref{lemma:proof-of-stability-of-theta}), the proof follows.
\end{proof}

\section{Attack design via direct optimization}\label{section:attack-design-via-KKT}

In this section, we will apply the well-known Karush-Kuhn-Tucker (KKT) conditions to find $\{\bm{T}_k^*,\bm{U}_k^*, \bm{M}_k^*, \bm{d}_k^*\}_{1 \leq k \leq N}$ for designing the attack at time $t$. 

\subsection{KKT based solution: the LAADE-KKT algorithm}
\label{subsection:LAADE-KKT}
Let us consider the modified constrained problem:
\begin{align}\label{eqn:modified-constrained-problem-1}
 && \min_{ \{\bm{T}_k^*,\bm{U}_k^*, \bm{M}_k^*, \bm{d}_k^*\}_{1 \leq k \leq N} } \sum_{k=1}^N \mathbb{E} ( ||\bm{\theta}^{(k)}(t)||^2  | \mathcal{F}_{t-1}) \nonumber\\
 && \textit{s.t.} \sum_{k=1}^N \mathbb{E}  (\tilde{\bm{z}}_k(t)' \bm{\Sigma}_k^{-1} \tilde{\bm{z}}_k(t)| \mathcal{F}_{t-1} ) \leq \frac{\alpha \eta}{J}\tag{MCP1}
\end{align}
Clearly, applying KKT conditions on the relaxed version of this problem, using a Lagrange multiplier $\lambda$, will involve setting the gradient of $f_t(\{\bm{T}_k, \bm{U}_k, \bm{M}_k, \bm{d}_k\}_{1 \leq k \leq N}, \lambda)$ w.r.t. the primal variables $\{\bm{T}_k, \bm{U}_k, \bm{M}_k, \bm{d}_k\}_{1 \leq k \leq N}$ equal to $\bm{0}$. 
However, it turns out that, the function  $f_t(\{\bm{T}_k, \bm{U}_k, \bm{M}_k, \bm{d}_k\}_{1 \leq k \leq N}, \lambda)$ is convex (by Lemma~\ref{lemma:convexity}) but not strictly convex w.r.t. $\{\bm{M}_k,\bm{d}_k\}_{1 \leq k \leq N}$, and that the derivative of this function w.r.t. $\{\bm{M}_k,\bm{d}_k\}_{1 \leq k \leq N}$ is a function of $\{\bm{M}_k \bm{\theta}^{(k)}(t-1)+\bm{d}_k\}_{1 \leq k \leq N}$, which can lead to many possible solutions. Hence, we introduce a regularization term involving   the Frobenius norm of $\{\bm{M}_k\}_{1 \leq k \leq N}$: 

\footnotesize
\begin{align}\label{eqn:modified-constrained-problem}
 && \min_{ \{\bm{T}_k^*,\bm{U}_k^*, \bm{M}_k^*, \bm{d}_k^*\}_{1 \leq k \leq N} } \sum_{k=1}^N \mathbb{E} ( ||\bm{\theta}^{(k)}(t)||^2  | \mathcal{F}_{t-1}) + \xi \sum_{k=1}^N ||\bm{M}_k||_F^2 \nonumber\\
 && \textit{s.t.} \sum_{k=1}^N \mathbb{E}  (\tilde{\bm{z}}_k(t)' \bm{\Sigma}_k^{-1} \tilde{\bm{z}}_k(t)| \mathcal{F}_{t-1} ) \leq \frac{\alpha \eta}{J}\tag{MCP}
\end{align}
\normalsize

where $\xi>0$ is a pre-determined constant. Applying KKT conditions on the relaxed version of \eqref{eqn:modified-constrained-problem}, using a Lagrange multiplier $\lambda$, will involve setting the gradient of $f_t(\{\bm{T}_k, \bm{U}_k, \bm{M}_k, \bm{d}_k\}_{1 \leq k \leq N}, \lambda)+ \xi \sum_{k=1}^N ||\bm{M}_k||_F^2$ w.r.t. the primal variables $\{\bm{T}_k, \bm{U}_k, \bm{M}_k, \bm{d}_k\}_{1 \leq k \leq N}$ equal to $\bm{0}$. 
This yields a set of linear equations \eqref{eqn:derivative-wrt-T}, \eqref{eqn:derivative-wrt-U}, \eqref{eqn:derivative-wrt-M}, \eqref{eqn:derivative-wrt-d} of these primal variables.  
\begin{lemma}
The optimal solution of \eqref{eqn:modified-constrained-problem} yields $\bm{U}_k^*=\bm{0}$ and hence $\bm{S}_k^*=\bm{0}$ for all $1 \leq k \leq N$.
\end{lemma}
\begin{proof}
 \eqref{eqn:derivative-wrt-U} directly shows that $\bm{U}_k^*=\bm{0}$, since $\bm{G}_k' \bm{G}_k + \lambda \bm{\Sigma}_k^{-1}$ is a positive definite matrix.  
\end{proof}
Hence, by solving \eqref{eqn:derivative-wrt-T},  \eqref{eqn:derivative-wrt-M} and \eqref{eqn:derivative-wrt-d}, we can find $\{\bm{T}_k^*(\lambda) ,   \bm{M}_k^*(\lambda), \bm{d}_k^*(\lambda)\}_{1 \leq k \leq N}$ as a function of $\lambda$. Putting these values in the constraint of \eqref{eqn:modified-constrained-problem} and equating both sides yields $\lambda$; then $\{\bm{T}_k^*(\lambda) ,   \bm{M}_k^*(\lambda), \bm{d}_k^*(\lambda)\}_{1 \leq k \leq N}$ can be used for the attack at time $t$. It is important to note that, $\{\bm{T}_k^*(\lambda) ,   \bm{M}_k^*(\lambda), \bm{d}_k^*(\lambda)\}_{1 \leq k \leq N}$ depend on the estimates, and thus on the history of observations as well. 

Note that, \eqref{eqn:modified-constrained-problem} is a quadratically constrained quadratic problem (QCQP), which is not necessarily convex. Hence, KKT conditions may not yield the globally  optimal solution. However, for the special case where $\{\bm{T}_k\}_{1 \leq k \leq N}$ is fixed, \eqref{eqn:modified-constrained-problem} becomes a convex optimization problem by Lemma~\ref{lemma:convexity}, and hence the above KKT-based procedure yields globally optimally solution. This algorithm is called {\em linear attack algorithm for distributed estimation based on KKT (LAADE-KKT)}.

\begin{figure*}
\hrule

\footnotesize
{\bf Differentiation w.r.t. $\bm{T}_k$:}
\begin{eqnarray} \label{eqn:derivative-wrt-T}
&&\bm{G}_k' \bm{G}_k \bm{T}_k^* \bigg [ \bm{H}_k \bm{A} \bm{\theta}^{(k)}(t-1) \bigg(\bm{\theta}^{(k)}(t-1)\bigg)' \bm{A}' \bm{H}_k' + \bm{H}_k \bm{Q} \bm{H}_k' + \bm{R}_k + 
\bm{H}_k \bm{A} \bigg( \bm{R}(t-1)+ (\hat{\bm{x}}(t-1)-\bm{x}^*) (\hat{\bm{x}}(t-1)-\bm{x}^*)'\bigg) \bm{A}' \bm{H}_k' \nonumber\\
&& -\bigg(\bm{H}_k \bm{A} \bm{\theta}^{(k)}(t-1)(\bm{\hat{x}}(t-1)-\bm{x}^*)'\bm{A}' \bm{H}_k'+
\bm{H}_k \bm{A} (\bm{\hat{x}}(t-1)-\bm{x}^*) (\bm{\theta}^{(k)}(t-1))'\bm{A}' \bm{H}_k' \bigg) \bigg]\nonumber\\
&&- \bm{G}_k' \bigg[ (\bm{A}-N_k \bm{C}_k \bm{A})\bm{\theta}^{(k)}(t-1)+ \bm{C}_k \bm{A} \sum_{j \in \mathcal{N}_k} \bm{\theta}^{(j)}(t-1)-(\bm{I}-\bm{A})\bm{x}^*+ \bm{G}_k (\bm{M}_k^* \bm{\theta}^{(k)}(t-1)+\bm{d}_k^*) \bigg]     (\bm{\theta}^{(k)}(t-1))'\bm{A}' \bm{H}_k'\nonumber\\
&& + \lambda \bm{\Sigma}_k^{-1} \bm{T}_k^* \bigg[ \bm{H}_k \bm{Q} \bm{H}_k' + \bm{R}_k + \bm{H}_k \bm{A} \bm{R}(t-1) \bm{H}_k' \bm{A}' + \bm{H}_k \bm{A} \bigg( \bm{\hat{x}}(t-1)-\bm{\hat{x}}^{(k)}(t-1) \bigg)  \bigg( \bm{\hat{x}}(t-1)-\bm{\hat{x}}^{(k)}(t-1) \bigg)' \bm{H}_k' \bm{A}' \bigg]\nonumber\\
&&+\lambda \bm{\Sigma}_k^{-1} \bigg( \bm{M}_k^* \bm{\theta}^{(k)}(t-1) + \bm{d}_k^* \bigg) \bigg( \hat{\bm{x}}(t-1)-\hat{\bm{x}}^{(k)}(t-1) \bigg) \bm{A}' \bm{H}_k'=0  
\end{eqnarray}
\normalsize
\hrule

\hrule
\footnotesize
{\bf Differentiation w.r.t. $\bm{U}_k$:}
\begin{equation}\label{eqn:derivative-wrt-U}
    \bigg( \bm{G}_k' \bm{G}_k + \lambda \bm{\Sigma}_k^{-1} \bigg) \bm{U}_k=0  
\end{equation}
\normalsize
\hrule 

\footnotesize
{\bf Differentiation w.r.t. $\bm{M}_k$:}
\begin{eqnarray}\label{eqn:derivative-wrt-M}
&& \bm{G}_k' \bm{G}_k \bm{M}_k^* \bm{\theta}^{(k)}(t-1)\bm(\bm{\theta}^{(k)}(t-1) \bm)' \nonumber\\
&& + 2 \bm{G}_k' \bigg( (\bm{A}-\bm{G}_k \bm{T}_k^* \bm{H}_k \bm{A}-N_k \bm{C}_k \bm{A})\bm{\theta}^{(k)}(t-1)+ \bm{C}_k \bm{A} \sum_{j \in \mathcal{N}_k} \bm{\theta}^{(j)}(t-1)-(\bm{I}-\bm{A})\bm{x}^*+ \bm{G}_k \bm{d}_k^* \bigg) \bigg(\bm{\theta}^{(k)}(t-1) \bigg)' \nonumber\\
&& + \bm{G}_k' \bm{G}_k \bm{T}_k^* \bm{H}_k \bm{A} \bigg(\hat{\bm{x}}(t-1)-\bm{x}^* \bigg) \bigg(\bm{\theta}^{(k)}(t-1)\bigg)'  + 2 \lambda \bm{\Sigma}_k^{-1} \bm{T}_k^* \bm{H}_k \bm{A} \bigg(  \hat{\bm{x}}(t-1)-  \hat{\bm{x}}^{(k)}(t-1)  \bigg) \bigg(\bm{\theta}^{(k)}(t-1)\bigg)' \nonumber\\
&& + 2 \lambda \bm{\Sigma}_k^{-1} \bm{M}_k^* \bm{\theta}^{(k)}(t-1) \bigg(\bm{\theta}^{(k)}(t-1)\bigg)'+ 2 \xi \bm{M}_k=0 
\end{eqnarray}
\normalsize
\hrule 

\footnotesize
{\bf Differentiation w.r.t. $\bm{d}_k$:}
\begin{eqnarray}\label{eqn:derivative-wrt-d}
&&  \bm{G}_k' \bigg( (\bm{A}-\bm{G}_k \bm{T}_k^* \bm{H}_k \bm{A}-N_k \bm{C}_k \bm{A})\bm{\theta}^{(k)}(t-1)+ \bm{C}_k \bm{A} \sum_{j \in \mathcal{N}_k} \bm{\theta}^{(j)}(t-1)-(\bm{I}-\bm{A})\bm{x}^*+ \bm{G}_k (\bm{M}_k^* \bm{\theta}^{(k)}(t-1)+\bm{d}_k^*) \bigg)\nonumber\\
&& + \bm{G}_k' \bm{G}_k \bm{T}_k^* \bm{H}_k \bm{A} \bigg(\hat{\bm{x}}(t-1)-\bm{x}^* \bigg) +  \lambda \bm{\Sigma}_k^{-1} \bigg( \bm{T}_k^* \bm{H}_k \bm{A} \hat{\bm{x}}(t-1)- \bm{T}_k^* \bm{H}_k \bm{A} \hat{\bm{x}}^{(k)}(t-1)+\bm{M}_k^* \bm{\theta}^{(k)}(t-1)+\bm{d}_k^*  \bigg) =0 
\end{eqnarray}
\normalsize
\hrule
 
\end{figure*}

\subsection{Updating $\lambda(t)$ iteratively: OLAADE-KKT}
Note that, solving \eqref{eqn:constrained-optimization-problem} will require us to solve a constrained average-cost Markov decision process (MDP; see \cite{bertsekas07dynamic-programming-optimal-control-2}) to find an optimal policy, since the  decision obtained by solving  \eqref{eqn:modified-constrained-problem} at any time will affect the future estimates made at the nodes, and thus the future cost incurred by the attacker as well. Obviously, solving \eqref{eqn:modified-constrained-problem} will always return a {\em myopic policy}. However, due to the complicated structure of the problem, especially due to the complex process of   evolution of the single stage objective function and constraint function in \eqref{eqn:constrained-optimization-problem} over time, we   resorted to solve \eqref{eqn:modified-constrained-problem} as an alternative to solving MDP. However, 
\eqref{eqn:modified-constrained-problem} is a one-shot optimization problem where the objective and constraint both are some conditional expectations given the history $\mathcal{F}_{t-1}$, while \eqref{eqn:constrained-optimization-problem} is a sequential optimization problem where the objective and constraint are averaged over independent sample paths. 

In this subsection, we will provide an online version of LAADE-KKT, i.e., OLAADE-KKT, which will seek to meet the constraint of \eqref{eqn:constrained-optimization-problem}. This algorithm maintains a running iterate  $\lambda(t-1)$, and computes $\bm{T}_k(t-1)=\bm{T}_k^*(\lambda(t-1)), \bm{M}_k(t-1)=\bm{M}_k^*(\lambda(t-1)), \bm{d}_k(t-1)=\bm{d}_k^*(\lambda(t-1))$ to solve \eqref{eqn:unconstrained-optimization-problem} at time $t$ by using the set of linear equations \eqref{eqn:derivative-wrt-T}, \eqref{eqn:derivative-wrt-U}, \eqref{eqn:derivative-wrt-M}, \eqref{eqn:derivative-wrt-d}. Then it makes the following update:
\begin{equation}\label{eqn:lambda-update-OLAADE-KKT-1}
 \lambda(t)=[\lambda(t-1)+b(t) (\sum_{k=1}^N \tilde{\tilde{\bm{z}}}_k(t)' \bm{\Sigma}_k^{-1} \tilde{\tilde{\bm{z}}}_k(t)-\frac{\alpha \eta}{J})]_0^{A_0}
\end{equation}
{\em where $\tilde{\tilde{\bm{z}}}_k(t)$ is the innovation at node~$k$ at time~$t$, which is obtained by applying $\{\bm{T}_k(t-1)=\bm{T}_k^*(\lambda(t-1)), \bm{M}_k(t-1)=\bm{M}_k^*(\lambda(t-1)), \bm{d}_k(t-1)=\bm{d}_k^*(\lambda(t-1))\}_{1 \leq k \leq N}$   on an independently generated/simulated state-observation sequence $\{\tilde{\tilde{\bm{x}}}(\tau),\tilde{\tilde{\bm{y}}}(\tau)\}_{0 \leq \tau \leq t}$.} Step size sequence  $\{b(t)\}_{t \geq 0}$ is a   sequence of non-negative numbers such that $\sum_{t=0}^{\infty}b(t)=\infty, \sum_{t=0}^{\infty}b^2(t)<\infty$. The iterations are projected onto a compact interval $[0,A_0]$ to ensure boundedness. The number $A_0$ is chosen to be sufficiently large so that, if,  for any $\lambda^* \geq 0$, the constraint in \eqref{eqn:modified-constrained-problem} is met with equality under $\{\bm{T}_k^*(\lambda^*), \bm{M}_k^*(\lambda^*), \bm{d}_k^*(\lambda^*)\}_{1 \leq k \leq N}$, then $\lambda^* \in [0,A_0)$.  This iteration is motivated by the theory of stochastic approximation~\cite{borkar08stochastic-approximation-book}, where the goal is to meet the constraint in \eqref{eqn:constrained-optimization-problem} with equality. This algorithm is referred to as OLAADE-KKT-1.

However, the constraint in \eqref{eqn:constrained-optimization-problem}   actually involves an upper bound to the attack detection probability averaged over   time. If the attacker has   access to the alarms raised by the detectors deployed in various nodes, then that additional information can be used to update $\lambda(t)$. Let the indicator that at least one alarm is raised  at time~$t$ be denoted by $I_t'$, {\em which is obtained by applying  $\{\bm{T}_k(t-1)=\bm{T}_k^*(\lambda(t-1)), \bm{M}_k(t-1)=\bm{M}_k^*(\lambda(t-1)), \bm{d}_k(t-1)=\bm{d}_k^*(\lambda(t-1))\}_{1 \leq k \leq N}$ on  an independently generated/simulated state-observation sequence $\{\tilde{\tilde{\bm{x}}}(\tau),\tilde{\tilde{\bm{y}}}(\tau)\}_{0 \leq \tau \leq t}$.} Then, $\lambda(t)$ can be updated as:
\begin{equation}\label{eqn:lambda-update-OLAADE-KKT-2}
 \lambda(t)=[\lambda(t-1)+b(t) (I_t'-\alpha)]_0^{A_0}
\end{equation}
Again here $A_0$ is chosen so large that, for any $\lambda^* \geq 0$ such that the detection probability $P_d(\lambda^*)=\alpha$, we have $\lambda^*<A_0$.

This modified algorithm is called OLAADE-KKT-2. It is interesting to note that OLAADE-KKT-2 is agnostic to the value of $\eta$ used by the detectors.

\subsubsection{Complexity reduction}\label{subsubsection:low-complexity-OLAADE-KKT}
Note that, in OLAADE-KKT-1, $\tilde{\tilde{\bm{z}}}_k(t)$ is the innovation at node~$k$ at time~$t$, when $\{\bm{T}_k^*(\lambda(t-1)), \bm{M}_k^*(\lambda(t-1)), \bm{d}_k^*(\lambda(t-1))\}_{1 \leq k \leq N}$ is applied on an independently genereted/simulated state-observation sequence $\{\tilde{\tilde{\bm{x}}}(\tau),\tilde{\tilde{\bm{y}}}(\tau)\}_{0 \leq \tau \leq t}$. Using an independently generated/simulated state-observation sequence up to time $t$ is necessary for the convergence proof of OLAADE-KKT-1, because a particular noise sequence in the convergence proof need to be Martingale difference noise sequence. Also, at each time $t$, we need to run this operation over the simulated history over time $\{0,1,\cdots,t\}$ in order to ensure that an offset term in the proof remains $o(1)$ instead of $O(1)$. Hence, computing $\{\tilde{\tilde{\bm{z}}}_k(t)\}_{1 \leq k \leq N}$ will require $O(t)$ computations at time~$t$, which is not practically feasible. However, we can avoid this $O(t)$ computation by replacing  $\tilde{\tilde{\bm{z}}}_k(t)$ in \eqref{eqn:lambda-update-OLAADE-KKT-1}  simply by $\tilde{\bm{z}}_k(t)$ which is the innovation at node~$k$ at time~$t$ under the scheme that applies $\{\bm{T}_k^*(\lambda(\tau-1)), \bm{M}_k^*(\lambda(\tau-1)), \bm{d}_k^*(\lambda(\tau-1))\}_{1 \leq k \leq N}$ on $\bm{y}(\tau)$ for all $\tau$. This low complexity version of OLAADE-KKT-1 is denoted by OLAADE-KKT-1-LC. 

Similarly, the $O(t)$ computation at time~$t$ for OLAADE-KKT-2 can be avoided by replacing $I_t'$ in \eqref{eqn:lambda-update-OLAADE-KKT-2} by $I_t$ which is obtained by applying $\{\bm{T}_k^*(\lambda(\tau-1)), \bm{M}_k^*(\lambda(\tau-1)), \bm{d}_k^*(\lambda(\tau-1))\}_{1 \leq k \leq N}$ on $\bm{y}(\tau)$ for all $\tau$; this low complexity version is henceforth called OLAADE-KKT-2-LC.

While the low-complexity versions are practically feasible, their convergence proof is technically very challenging. Hence, we will only prove convergence of OLAADE-KKT-1 and OLAADE-KKT-2 later in this paper.

\subsection{Convergence analysis of OLAADE-KKT}
Since LAADE-KKT does not involve any iteration, it does not exhibit any convergence property. Here, we discuss convergence properties of OLAADE-KKT-1 and OLAADE-KKT-2, where $\{\bm{T}_k\}_{1 \leq k \leq N}$ are fixed and known, so that  \eqref{eqn:modified-constrained-problem} becomes a convex optimization problem by Lemma~\ref{lemma:convexity}.

\begin{assumption}\label{assumption:M-has-max-eigenvalue-less-than-1}
The matrices $\{\bm{T}_k\}_{1 \leq k \leq N}$ are such that the $\bm{M}$ matrix of Section~\ref{section:error-dynamics-under-attack} has a spectral radius less than $1$.
\end{assumption}

\subsubsection{Convergence of OLAADE-KKT-1}\label{subsubsection:convergence-of-OLAADE-KKT-1}
Note that, if OLAADE-KKT-1 uses a fixed $\lambda \geq 0$ all the time, then at time~$t$, the attacker takes up   the history available   up to time $(t-1)$, and computes $\{\bm{M}_k^*(\lambda, \{\hat{\bm{x}}^{(j)}(t-1)\}_{1 \leq j \leq N}, \hat{\bm{x}}(t-1)),  \bm{d}_k^*(\lambda, \{\hat{\bm{x}}^{(j)}(t-1)\}_{1 \leq j \leq N}, \hat{\bm{x}}(t-1))\}_{1 \leq k \leq N}$ (which are sample-path-dependent, i.e., dependent on $\{\bm{y}(\tau)\}_{0 \leq \tau \leq t-1}$) which are further   used to compute the estimates at time $t$. 

\begin{lemma}\label{lemma:OLAADE-KKT-1-iterates-reach-steady-state-distribution}
For a fixed $\lambda \geq 0$ and under OLAADE-KKT-1 and Assumption~\ref{assumption:M-has-max-eigenvalue-less-than-1}, the sequence of iterates $\{\bm{M}_k(t), \bm{d}_k(t)\}_{1 \leq k \leq N, t \geq 0}$ reach a steady state distribution $g_{\lambda}^*(\cdot)$.
\end{lemma}
\begin{proof}
By Assumption~\ref{assumption:M-has-max-eigenvalue-less-than-1} and Lemma~\ref{lemma:proof-of-stability-of-theta}, $\{\hat{\bm{x}}_k(t)\}_{t \geq 0}$ and $\{\hat{\bm{x}}(t)\}_{t \geq 0}$ are stable. Hence, from \eqref{eqn:derivative-wrt-M} and \eqref{eqn:derivative-wrt-d}, the lemma is proved.
\end{proof}

Let us define the distribution of $\{\bm{M}_k(t), \bm{d}_k(t)\}_{1 \leq k \leq N, t \geq 0}$ under OLAADE-KKT-1 with a fixed $\lambda$ as $g_{t, \lambda}(\cdot)$, and the distribution of $\{\bm{M}_k(t), \bm{d}_k(t)\}_{1 \leq k \leq N, t \geq 0}$ under OLAADE-KKT-1 with   $\lambda(t)$ update as $g_t(\cdot)$. 
Also, let $\mu_{\lambda, \{\bm{M}_k, \bm{d}_k\}_{1 \leq k \leq N}}$ denote a generic   decision rule or policy    under  OLAADE-KKT-1 with a fixed parameter set $\lambda, \{\bm{M}_k, \bm{d}_k\}_{1 \leq k \leq N}$.  

Let us define:
\begin{eqnarray*}
    \Lambda & \doteq & \{\lambda \in [0,A_0) : \lim_{t \rightarrow \infty} \mathbb{E}_{\{\bm{M}_k, \bm{d}_k\}_{1 \leq k \leq N} \sim g_{\lambda}^*(\cdot) } \mathbb{E}_{\mu_{\lambda, \{\bm{M}_k, \bm{d}_k\}_{1 \leq k \leq N}}} \\
    && [\sum_{k=1}^N \tilde{\tilde{\bm{z}}}_k(t)' \bm{\Sigma}_k^{-1} \tilde{\tilde{\bm{z}}}_k(t)] =\frac{\alpha \eta}{J}\}
\end{eqnarray*}

\begin{theorem}\label{theorem:convergence-of-lambda-OLAADE-KKT-1}
Under Assumption~\ref{assumption:M-has-max-eigenvalue-less-than-1} and  OLAADE-KKT-1, the iterates $\lambda(t) \rightarrow \Lambda$  almost surely, and the limiting distributions satisfy $\lim_{t \rightarrow \infty}||g_t(\cdot)-g_{t,\lambda(t)}(\cdot)||_{TV}=0$ almost surely.
\end{theorem}
\begin{proof}
See Appendix~\ref{appendix:proof-of-convergence-of-lambda-OLAADE-KKT-1}. The proof is based on the theory of stochastic approximation in  \cite{borkar08stochastic-approximation-book}. 
\end{proof}
However, it is important to note that the convergence can be sample-path dependent.

\subsubsection{Convergence of OLAADE-KKT-2}
Let   us define:

\footnotesize
\begin{eqnarray*}
    \Lambda' & \doteq & \{\lambda \in [0,A_0) : \lim_{t \rightarrow \infty} \mathbb{E}_{\{\bm{M}_k, \bm{d}_k\}_{1 \leq k \leq N} \sim g_{\lambda}^*(\cdot) } \mathbb{E}_{\mu_{\lambda, \{\bm{M}_k, \bm{d}_k\}_{1 \leq k \leq N}}} \\
    && (I_t') =\alpha \}
\end{eqnarray*}
\normalsize

\begin{theorem}\label{theorem:convergence-of-lambda-OLAADE-KKT-2}
Under Assumption~\ref{assumption:M-has-max-eigenvalue-less-than-1} and  OLAADE-KKT-2, the iterates $\lambda(t) \rightarrow \Lambda'$  almost surely, and the limiting distributions satisfy $\lim_{t \rightarrow \infty}||g_t(\cdot)-g_{t,\lambda(t)}(\cdot)||_{TV}=0$ almost surely.
\end{theorem}
\begin{proof}
The proof is very similar to that of Theorem~\ref{theorem:convergence-of-lambda-OLAADE-KKT-2}, except that we use $I_t'$ instead of $\sum_{k=1}^N \tilde{\tilde{\bm{z}}}_k(t)' \bm{\Sigma}_k^{-1} \tilde{\tilde{\bm{z}}}_k(t)$ in this proof. Hence, we omit details of the proof.
\end{proof}

\section{Attack design via SPSA}
\label{section:attack-design-SPSA}
In this section, we propose an {\em  online linear attack algorithm for distributed estimation using SPSA} (OLAADE-SPSA) that allows us to avoid solving the KKT equations at each time~$t$. The OLAADE-SPSA algorithm involves two-timescale stochastic approximation \cite{borkar08stochastic-approximation-book}, which is basically a stochastic gradient descent algorithm with a noisy gradient estimate; \eqref{eqn:unconstrained-optimization-problem} is solved via SPSA in the faster timescale, and $\lambda$ is updated in the slower timescale. 

\subsection{Description of OLAADE-SPSA}
The algorithm requires three positive step size sequences $\{a(t)\}_{t \geq 0}$, $\{b(t)\}_{t \geq 0}$ and $\{c(t)\}_{t \geq 0}$ satisfying the following criteria: (i) $\sum_{t=0}^{\infty} a(t)=\sum_{t=0}^{\infty} b(t)=\infty$, 
(ii) $\sum_{t=0}^{\infty} a^2(t)<\infty, \sum_{t=0}^{\infty} b^2(t)<\infty$, 
(iii) $\lim_{t \rightarrow \infty} \frac{b(t)}{a(t)}=0$, (iv) $\lim_{t \rightarrow \infty} c(t)=0$, and (v) $\sum_{t=0}^{\infty} \frac{a^2 (t)}{c^2(t)}<\infty$. The first three conditions are standard requirements for two-timescale stochastic approximation. The fourth condition ensures that the gradient estimate is asymptotically unbiased, and the fifth condition is required for the convergence of SPSA.

\vspace{2mm}
\hrule
\noindent {\bf The OLAADE-SPSA algorithm}
 \hrule
 \vspace{0.5mm}
 \noindent {\bf Input:} $\{a(t)\}_{t \geq 0}$, $\{b(t)\}_{t \geq 0}$, $\{c(t)\}_{t \geq 0}$, $\alpha$, $\eta$, $J$, $A_0$.

\noindent {\bf Initialization:} $\bm{T}_k(0)$,   $\bm{M}_k(0)$, $\bm{d}_k(0)$ for all $k \in \mathcal{N}$, $\lambda(0), \{\hat{\bm{x}}^{(k)}(0)\}_{1 \leq k \leq N}$, $\hat{\bm{x}}(0)$

\noindent {\bf For $t=1,2,3,\cdots$:}

\begin{enumerate}
\item For each $1 \leq  k \leq N$, the attacker  generates  random matrices $\bm{\Delta}^{(k)}(t)$,   $\bm{\Pi}^{(k)}(t)$ and $\bm{\beta}^{(k)}(t)$ having same dimensions as $\bm{T}_k(t-1)$,   $\bm{M}_k(t-1)$ and $\bm{d}_k(t-1)$ respectively,  whose entries are uniformly and independently chosen from the set $\{-1,1\}$.

\item The attacker computes $\bm{T}_k^+ \doteq \bm{T}_k(t-1) + c(t) \bm{\Delta}^{(k)}(t)$, $\bm{T}_k^- \doteq \bm{T}_k(t-1) - c(t) \bm{\Delta}^{(k)}(t)$,  
$\bm{M}_k^+ \doteq \bm{M}_k(t-1) + c(t) \bm{\Pi}^{(k)}(t)$, 
$\bm{M}_k^- \doteq \bm{M}_k(t-1) - c(t) \bm{\Pi}^{(k)}(t)$, $\bm{d}_k^+ \doteq \bm{d}_k(t-1) + c(t) \bm{\beta}^{(k)}(t)$, $\bm{d}_k^- \doteq \bm{d}_k(t-1) - c(t) \bm{\beta}^{(k)}(t)$, for all $1 \leq k \leq N$.

\item The attacker computes:
\begin{eqnarray}
 \kappa_t^+ &\doteq& \sum_{j=1}^N \mathbb{E} \bigg( (||\bm{\theta}^{(j)}(t)||^2 + \lambda(t-1) \tilde{\bm{z}}_j(t)' \bm{\Sigma}_j^{-1} \tilde{\bm{z}}_j(t)\nonumber\\
 && +\xi ||\bm{M}_k^+||_F^2 | \mathcal{F}_{t-1}, \{\bm{T}_k^+,  \bm{M}_k^+, \bm{d}_k^+\}_{1 \leq k \leq N}   \bigg) \nonumber
\end{eqnarray} 
using  \eqref{eqn:variance-of-theta} and \eqref{eqn:variance-of-z} under $\{\bm{T}_k^+,  \bm{M}_k^+, \bm{d}_k^+\}_{1 \leq k \leq N}$. The attacker computes $\kappa_t^-$ in a similar way using $\{\bm{T}_k^-,   \bm{M}_k^-, \bm{d}_k^-\}_{1 \leq k \leq N}$.
\item The attacker updates each element $(i,j)$ of $\bm{T}_k(t-1)$,  $\bm{M}_k(t-1)$ and $\bm{d}_k(t-1)$ for all $1 \leq k \leq N$ as follows:

\footnotesize
\begin{eqnarray}
 \bm{T}_k(t)(i,j) &=& \bigg[ \bm{T}_k(t-1)(i,j)-a(t) \times \frac{(\kappa_t^+ - \kappa_t^-)}{2 c(t) \bm{\Delta}_{(i,j)}^{(k)}(t)} \bigg]_{-A_0}^{A_0} \nonumber\\
  \bm{M}_k(t)(i,j) &=& \bigg[ \bm{M}_k(t-1)(i,j)-a(t) \times \frac{(\kappa_t^+ - \kappa_t^-)}{2 c(t) \bm{\Pi}_{(i,j)}^{(k)}(t)}  \bigg]_{-A_0}^{A_0}\nonumber\\
  \bm{d}_k(t)(i,1) &=& \bigg[ \bm{d}_k(t-1)(i,1)-a(t) \times \frac{(\kappa_t^+ - \kappa_t^-)}{2 c(t) \bm{\beta}_{(i,1)}^{(k)}(t)}  \bigg]_{-A_0}^{A_0} \nonumber\\
\end{eqnarray}
\normalsize

\item The sensors make observations $\{\bm{y}_k(t)\}_{1 \leq k \leq N}$, which are accessed by the attacker.
\item The attacker calculates $\bm{z}_k(t)=\bm{y}_k(t)-\bm{H}_k \bm{A} \hat{\bm{x}}^{(k)}(t-1)$  for all $k \in \{1,2,\cdots,N\}$.
\item The attacker calculates  $\bm{\tilde{z}_k}(t)=\bm{T}_k(t) \bm{z}_k(t)+\bm{b}_k(t)$ for all $k \in \{1,2,\cdots,N\}$, where $\bm{b}_k(t) = \bm{M}_k(t) \bm{\theta}^{(k)}(t-1)+\bm{d}_k(t)$. The observations are accordingly modified as $\tilde{\bm{y}}_k(t)=\tilde{\bm{z}}_k(t)+\bm{H}_k \bm{A} \hat{\bm{x}}^{(k)}(t-1)$ and sent to the agent nodes.

\item  The attacker updates the Lagrange multiplier as follows:

 {\em \textbf{If $\eta$ is known to attacker: OLAADE-SPSA-1}}

\begin{equation}\label{eqn:lambda-update-OLAADE-SPSA-2a}
 \lambda(t)=[\lambda(t-1)+b(t) (\sum_{k=1}^N \tilde{\bm{z}}_k(t)' \bm{\Sigma}_k^{-1} \tilde{\bm{z}}_k(t)-\frac{\alpha \eta}{J})]_0^{A_0}
\end{equation}

{\em \textbf{If $\eta$ is unknown to attacker but alarms are observable: OLAADE-SPSA-2}}

\begin{equation}\label{eqn:lambda-update-OLAADE-SPSA-2b}
 \lambda(t)=[\lambda(t-1)+b(t) (I_t-\alpha)]_0^{A_0}
\end{equation}
\item The agent nodes compute the estimates locally, using \eqref{eqn:KCF-equation} and  the modified $\{\tilde{\bm{y}}_k(t)\}_{1 \leq k \leq N}$. The agent nodes broadcast their estimates to their neighboring nodes.
\end{enumerate}
\noindent {\bf end}
\label{algorithm:correction-algorithm-learning} 
 \hrule
\vspace{2mm}

\subsection{Discussion of OLAADE-SPSA}

\begin{enumerate}
 \item If $\{\bm{T}_k\}_{1 \leq k \leq N}$ is kept fixed, then the first update in step~$4$ of OLAADE-SPSA is not required.
 \item The OLAADE-SPSA algorithm combines the online stochastic gradient descent (OSGD)  algorithm of \cite[Chapter~$3$]{hazan2016introduction} with two-timescale stochastic approximation of \cite{borkar08stochastic-approximation-book}. The $\lambda(t)$ iterate is updated in the slower timescale   to meet either the constraint in \eqref{eqn:constrained-optimization-problem} or the exact attack detection probability  constraint with equality. In the faster timescale, OSGD is used for solving  \eqref{eqn:unconstrained-optimization-problem}. Since $\lim_{t \rightarrow \infty} \frac{b(t)}{a(t)}=0$, the faster timescale iterates $\{\bm{T}_k(t),  \bm{M}_k(t), \bm{d}_k(t)\}_{1 \leq k \leq N}$ view the slower timescale iterate $\lambda(t)$ as quasi-static, while the $\lambda(t)$ iteration finds the faster timescale iterates as almost equilibriated; as if, the faster timescale iterates are varied in an inner loop and the slower timescale iterate is varied in an outer loop.
 \item Steps~$1-4$ of OLAADE-SPSA is basically using SGD, but via simultaneous perturbation stochastic approximation (SPSA; see \cite{spall92original-SPSA}). SPSA allows us to avoid coordinate wise perturbation for gradient estimation of the function under consideration, by providing a zero-mean random  perturbation to all coordinates (entries) of a vector or matrix variable simultaneously and independently.  Steps~$1-4$ of OLAADE-SPSA is equivalent to one iteration of SGD by using SPSA, where the time-varying function to optimize is $  \sum_{k=1}^N \mathbb{E}(||\bm{\theta}^{(k)}(t)||^2 + \lambda(t-1) \tilde{\bm{z}}_k(t)' \bm{\Sigma}_k^{-1} \tilde{\bm{z}}_k(t) + \xi ||\bm{M}_k||_F^2| \mathcal{F}_{t-1} )$.
 \item All iterates are projected onto various large but compact intervals to ensure boundedness.
\end{enumerate}

\section{Numerical results}\label{section:numerical-work}

We consider a distributed system with $N=6$ agent nodes and consider two different network topologies, the 3-regular hexagon and the line topology. The underlying process is q-dimensional, with $q=2$, while the observations recorded at each node $y_k \in \mathbb{R}^3$.  The system parameters $\bm{A}, \bm{Q}, \{\bm{R}_k\}_{1 \leq k \leq 6}, \{\bm{H}_k\}_{1 \leq k \leq 6}$ are chosen randomly and independently for the two different topologies. The KCF parameters $\{\bm{G}_k, \bm{C}_k\}_{1 \leq k \leq 6}$ are computed using a technique from \cite{saber09kalman-consensus-optimality-stability}, and $\{\bm{\Sigma}_k\}_{1 \leq k \leq 6}$ are computed by simulating the KCF under no attack. 

 For FDI attack, we set $\bm{x}^*=[2, 2]'$, $\eta=300$, $\chi^2$  window size $J=10$ and $\lambda(0)=4$ and regularization constant $\zeta = 0.5$. To maintain the convexity of the problem, we fix $\bm{T}_k(t)=\bm{I}$, $\forall$ $1 \leq k \leq 6$ and $\forall$  $t \geq 1$. We then allow the algorithm to run until convergence on $\lambda(t)$. The $\chi^2$ detector raises concerns about FDI though alarms. 
 
 For the attack variants KKT-1 and SPSA-1, the adversary does not have access to the alarms. In this case, we notice that the Markov inequality based upper bound to the detection probability $P_d$ as in \eqref{eqn:Markov-bound} is too loose in practice, which in turn leads to a higher than necessary penalty in $\lambda$ update equation \eqref{eqn:lambda-update-OLAADE-KKT-1}. To alleviate  this problem, we introduce a hyper-parameter $c$ to be multiplied to the term $\frac{\alpha\eta}{J}$, which is tuned to get closer to the detection probability upper bound. 
For the KKT-2 and SPSA-2 variants, since the attacker has access to alarm triggers at the nodes, such a hyper-parameter is not required. 

Motivated by the ADAM algorithm   \cite{KingmaB14}, we implement an adaptive step size optimization variant for $\lambda(t)$ for faster convergence. However, to be able to reasonably observe the effect of changing $\lambda$ on the detection probability, we update $\lambda$ on a lower timescale of $0.1 \times$, i.e., for each iterative update of $\lambda$, we let the underlying process be simulated for $10$ iterations before the next update. 
 
\subsection{OLAADE-KKT}
 Recall that for OLAADE-KKT, we want to obtain the value of $\lambda$ for optimizing the MSE from target vs detection probability trade-off. Once the $\lambda(t)$ iterate converges to $\lambda^*$, we  simulate multiple sample paths under this fixed $\lambda^*$, and calculate the deviation from target, i.e.,   $\frac{1}{T}\sum_{t=1}^T \sum_{k=1}^N||\hat{\bm{x}}^{(k)}(t)-\bm{x}^*||^2$ for each sample path.
 
  In Figure~\ref{fig:avg_mse_performance}, we demonstrate the effectiveness of the attack along one sample path, by plotting the   deviation of the state estimates from the specified target across the  nodes, under  attack and no attack cases.   
  The broader simulation results for OLAAD-KKT-1-LC are summarized in Table~\ref{table:table6} and Table~\ref{table:table7}. The mean and standard deviation values are obtained from 10 sample runs. Similar results for OLAAD-KKT-2-LC are summarized in Table~\ref{table:table8} and Table~\ref{table:table9}. For OLAAD-KKT-1-LC, we report the results for that particular choice of hyper-parameter which allowed us to achieve the detection probability closest to the target, based on a  grid-search. 

 \begin{figure}[h!]
 \begin{centering}
 \begin{center}
 \includegraphics[height=3.8cm, width=9cm]{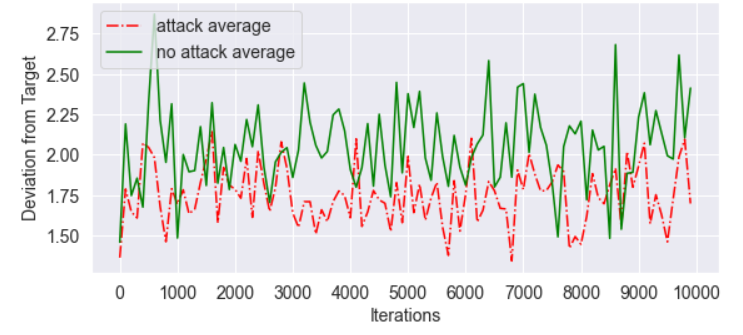}
 \vspace{-5mm}
 \caption{OLAADE-KKT-1-LC: Average MSE from $x^*$, 3-regular topology, $\alpha = 0.3$}
 \end{center}
 \end{centering}
 \label{fig:avg_mse_performance}
 \vspace{-5mm}
 \end{figure}

\begin{table}[h!]
\footnotesize
\centering
\scalebox{0.8}{\begin{tabular}{|c |c |c |c |c |}
\hline
Permissible   &  Detection &  Detection & Deviation & Deviation    \\
 detection  & probability & probability  & from $\bm{x}^*$  & from $\bm{x}^*$  \\ 
 probability ($\alpha$) & (no attack) & under FDI &  (no attack) & under FDI  \\
 \hline
 & & & & \\
0.2 & 0.044 +/- 0.003 & 0.186 +/- 0.01 & 2.062 +/- 0.002 & 1.793 +/- 0.002 \\
0.3 & 0.044 +/- 0.005 & 0.286 +/- 0.011 & 2.063 +/- 0.004 & 1.738 +/- 0.003 \\
\hline
 \end{tabular}}
\normalsize
\caption{$N=6$, 3-regular topology, OLAADE-KKT-1-LC}
\label{table:table6}
\end{table}

\begin{table}[h!]
\footnotesize
\centering
\scalebox{0.8}{\begin{tabular}{|c |c |c |c |c |}
\hline
  Permissible   &  Detection &  Detection & Deviation & Deviation    \\
 detection  & probability & probability  & from $\bm{x}^*$  & from $\bm{x}^*$  \\ 
 probability($\alpha$) & (no attack) & under FDI &  (no attack) & under FDI  \\
 \hline
 & & & & \\
0.25 & 0.047 +/- 0.004 & 0.235 +/- 0.009 & 2.045 +/- 0.007 & 1.289 +/- 0.003 \\
0.4 & 0.05 +/- 0.005 & 0.389 +/- 0.013 & 2.038 +/- 0.008 & 1.196 +/- 0.003 \\
\hline
 \end{tabular}}
\normalsize
\caption{$N=6$, Line topology,  OLAADE-KKT-1-LC}
\label{table:table7}
\end{table}

\begin{table}[h!]
\footnotesize
\centering
\scalebox{0.8}{\begin{tabular}{|c |c |c |c |c |}
\hline
Permissible   &  Detection &  Detection & Deviation & Deviation    \\
 detection  & probability & probability  & from $\bm{x}^*$  & from $\bm{x}^*$  \\ 
 probability ($\alpha$) & (no attack) & under FDI &  (no attack) & under FDI  \\
 \hline
 & & & & \\
0.2 & 0.046 +/- 0.003 & 0.178 +/- 0.008 & 2.063 +/- 0.005 & 1.799 +/- 0.004 \\
0.3 & 0.044 +/- 0.003 & 0.287 +/- 0.013 & 2.062 +/- 0.002 & 1.741 +/- 0.002 \\
\hline
 \end{tabular}}
\normalsize
\caption{$N=6$, 3-regular topology, OLAADE-KKT-2-LC}
\label{table:table8}
\end{table}

\begin{table}[h!]
\footnotesize
\centering
\scalebox{0.8}{\begin{tabular}{|c |c |c |c |c |}
\hline
  Permissible   &  Detection &  Detection & Deviation & Deviation    \\
 detection  & probability & probability  & from $\bm{x}^*$  & from $\bm{x}^*$  \\ 
 probability($\alpha$) & (no attack) & under FDI &  (no attack) & under FDI  \\
 \hline
 & & & & \\
0.25 &  0.05 +/- 0.004 & 0.223 +/- 0.01 & 2.048 +/- 0.012 & 1.305 +/- 0.005 \\
0.4 & 0.049 +/- 0.003 & 0.355 +/- 0.016 & 2.046 +/- 0.009 & 1.211 +/- 0.004 \\
\hline
 \end{tabular}}
\normalsize
\caption{$N=6$, Line topology,  OLAADE-KKT-2-LC}
\label{table:table9}
\end{table}

As mentioned previously, it is important to note that the underlying process parameters were different for the two topologies. This can be seen from the fact that the detection probability under the no-attack case varies slightly for the two settings. In fact, the nature of these underlying parameters often determines how well the attack can drive the estimates to the target value, while keeping the detection rate under $\alpha$. 

\subsection{OLAADE-SPSA}

We repeat the same set of experiments, with the same set of attack parameters for the OLAADE-SPSA attack scheme. Note that in this case, we want to estimate the values of $M$, $d$ for mounting an effective attack. As before, we report the mean performance of the attack, averaged over ten sample runs. It is again observed that OLAADE-SPSA is able to push all estimates closer to the target, while respecting the detection constraint.

\begin{table}[h!]
\footnotesize
\centering
\scalebox{0.8}{\begin{tabular}{|c |c |c |c |c |}
\hline
  Permissible   &  Detection &  Detection & Deviation & Deviation    \\
 detection  & probability & probability  & from $\bm{x}^*$  & from $\bm{x}^*$  \\ 
 probability ($\alpha$) & (no attack) & under FDI &  (no attack) & under FDI  \\
 \hline
 & & & & \\
0.2 & 0.043 +/- 0.005 & 0.189 +/- 0.012 & 2.062 +/- 0.003 & 1.804 +/- 0.003 \\
0.3 & 0.044 +/- 0.005 & 0.28 +/- 0.013 & 2.062 +/- 0.002 & 1.75 +/- 0.002 \\
\hline
 \end{tabular}}
\normalsize
\caption{$N=6$, 3-regular topology, OLAADE-SPSA-1}
\label{table:table10}
\end{table}

\begin{table}[h!]
\footnotesize
\centering
\scalebox{0.8}{\begin{tabular}{|c |c |c |c |c |}
\hline
  Permissible   &  Detection &  Detection & Deviation & Deviation    \\
  detection  & probability & probability  & from $\bm{x}^*$  & from $\bm{x}^*$  \\ 
 probability($\alpha$) & (no attack) & under FDI &  (no attack) & under FDI  \\
 \hline
 & & & & \\
0.25 & 0.052 +/- 0.006 & 0.232 +/- 0.014 & 2.042 +/- 0.014 & 1.491 +/- 0.010 \\
0.4 & 0.049 +/- 0.006 & 0.385 +/- 0.008 & 2.041 +/- 0.009 & 1.381 +/- 0.005 \\
\hline
 \end{tabular}}
\normalsize
\caption{$N=6$, Line topology,  OLAADE-SPSA-1}
\label{table:table11}
\end{table}

\begin{table}[h!]
\footnotesize
\centering
\scalebox{0.8}{\begin{tabular}{|c |c |c |c |c |}
\hline
Permissible   &  Detection &  Detection & Deviation & Deviation    \\
 detection  & probability & probability  & from $\bm{x}^*$  & from $\bm{x}^*$  \\ 
 probability ($\alpha$) & (no attack) & under FDI &  (no attack) & under FDI  \\
 \hline
 & & & & \\
0.2 & 0.043 +/- 0.004 & 0.184 +/- 0.01 & 2.064 +/- 0.003 & 1.805 +/- 0.002 \\
0.3 & 0.045 +/- 0.006 & 0.292 +/- 0.013 & 2.061 +/- 0.004 & 1.746 +/- 0.004 \\
\hline
 \end{tabular}}
\normalsize
\caption{$N=6$, 3-regular topology, OLAADE-SPSA-2}
\label{table:table12}
\end{table}

\begin{table}[h!]
\footnotesize
\centering
\scalebox{0.8}{\begin{tabular}{|c |c |c |c |c |}
\hline
   Permissible   &  Detection &  Detection & Deviation & Deviation    \\
  detection  & probability & probability  & from $\bm{x}^*$  & from $\bm{x}^*$  \\ 
 probability($\alpha$) & (no attack) & under FDI &  (no attack) & under FDI  \\
 \hline
 & & & & \\
 0.25 & 0.049 +/- 0.006 & 0.234 +/- 0.011 & 2.04 +/- 0.009 & 1.426 +/- 0.005 \\
 0.4 & 0.054 +/- 0.007 & 0.385 +/- 0.012 & 2.042 +/- 0.012 & 1.323 +/- 0.007 \\
\hline
 \end{tabular}}
\normalsize
\caption{$N=6, q=2$, Line topology,  OLAADE-SPSA-2}
\label{table:table13}
\end{table}

\subsection{Discussion}
We highlight some key takeaways from the simulation results. Firstly, the OLAADE-KKT attack variants are always at least as good or better than their OLAADE-SPSA counterparts, depending on the underlying process parameters. This matches our intuition, since the KKT variants are indeed provably optimal for the convex formulation. However, it is important to note that the KKT algorithms require us to solve a family of matrix equations at each iteration, which requires matrix inversion; this makes the computational complexity of the KKT variants per slot higher than that of the SPSA variants.

The second observation is that, the performance of the respective variants of KKT and SPSA when the adversary does not have direct access to alarms does not alter much even if access is made available. In practice,  however, this will seldom be the case, since the true values of $\eta$, $J$ and $\alpha$ are not directly available to the attacker apriori, and will therefore need to be assumed. Therefore, any conservative attacker without access to alarms would tend to lower the estimate for the detection threshold in order to avoid detection, and consequently, the performance of the attack without access to alarms will be worse.

\section{Conclusions}\label{section:conclusion}
In this paper, we  designed an optimal linear attack for distributed cyber-physical systems. The problem was posed an a constrained optimization problem.  The parameters of the attack scheme were learnt and optimized on-line, using tools from KKT, two-timescale stochastic approximation and SPSA. Numerical results demonstrated the efficacy of each of the proposed attack scheme. It is important to note that OLAAD-KKT based attacks require an active adversary in the sense that while the attack parameters converge in a distribution, they have to be updated in each iteration to remain effective. And while OLAAD-SPSA does not have that particular bottleneck, it can often require more effort to tune its parameters for convergence. In future, we seek to extend this work for unknown process and observation dynamics, and also prove convergence of the proposed algorithms.

\appendices

\section{Proof of Theorem~\ref{theorem:theta-and-z-evolution}}\label{appendix:proof-of-theta-and-z-evolution}
Under this FDI attack, we have:

\footnotesize
\begin{eqnarray}
 &&\hat{\bm{x}}^{(k)}(t) \nonumber\\
 &=& \bm{A} \hat{\bm{x}}^{(k)}(t-1)+ \bm{G}_k \tilde{\bm{z}}_k(t)
+ \bm{C}_k \sum_{j \in \mathcal{N}_k} (\bar{\bm{x}}^{(j)}(t)-\bar{\bm{x}}^{(k)}(t)) \nonumber\\
 &=& \bm{A} \hat{\bm{x}}^{(k)}(t-1)+ \bm{G}_k (\bm{T}_k(\bm{y}_k(t)-\bm{H}_k \bm{A} \hat{\bm{x}}^{(k)}(t-1))+\bm{b}_k(t)) \nonumber\\
 &+& \bm{C}_k \bm{A} \sum_{j \in \mathcal{N}_k} (\hat{\bm{x}}^{(j)}(t-1)-\hat{\bm{x}}^{(k)}(t-1)) 
\label{eqn_pd}
\end{eqnarray}
\normalsize

 Now, 

\footnotesize
\begin{eqnarray}
 &&\bm{\theta}^{(k)}(t) \nonumber\\
 &=& (\bm{A}-\bm{G}_k \bm{T}_k \bm{H}_k \bm{A}) \hat{\bm{x}}^{(k)}(t-1) \nonumber\\
 && + \bm{G}_k \bm{T}_k \underbrace{\bm{y}_k(t)}_{\doteq \bm{H}_k \bm{A}\bm{x}(t-1)+ \bm{H}_k \bm{w}(t-1)+\bm{v}_k(t)}+ \bm{G}_k \bm{b}_k(t)\nonumber\\
&& + \bm{C}_k \bm{A} \sum_{j \in \mathcal{N}_k} (\hat{\bm{x}}^{(j)}(t-1)-\hat{\bm{x}}^{(k)}(t-1))-\bm{x}^*\nonumber\\
 &=& (\bm{A}-\bm{G}_k \bm{T}_k \bm{H}_k \bm{A}) \bm{\theta}^{(k)}(t-1)+ \bm{G}_k \bm{T}_k \bm{H}_k \bm{A}\bm{\phi}(t-1) \nonumber\\
 && + \bm{C}_k \bm{A} \sum_{j \in \mathcal{N}_k} (\bm{\theta}^{(j)}(t-1)-\bm{\theta}^{(k)}(t-1))-(\bm{I}-\bm{A})\bm{x}^* \nonumber\\
 && + \bm{G}_k \bm{T}_k \bm{H}_k \bm{w}(t-1) + \bm{G}_k \bm{b}_k(t) + \bm{G}_k \bm{T}_k \bm{v}_k(t) \nonumber\\
  &=& (\bm{A}-\bm{G}_k \bm{T}_k \bm{H}_k \bm{A}-N_k \bm{C}_k \bm{A}) \bm{\theta}^{(k)}(t-1) \nonumber\\
 && + \bm{C}_k \bm{A} \sum_{j \in \mathcal{N}_k} \bm{\theta}^{(j)}(t-1)-(\bm{I}-\bm{A})\bm{x}^* + \bm{G}_k \bm{T}_k \bm{H}_k \bm{A}\bm{\phi}(t-1) \nonumber\\
 && + \bm{G}_k \bm{T}_k \bm{H}_k \bm{w}(t-1) + \bm{G}_k \bm{b}_k(t) + \bm{G}_k \bm{T}_k \bm{v}_k(t) \label{eqn:theta-evolution}
\end{eqnarray}
\normalsize

Clearly, $\mathbb{E}(||\bm{\theta}^{(k)}(t)||^2| \mathcal{F}_{t-1})$  can be expressed as  \eqref{eqn:variance-of-theta}; in this expression, we have used the fact that, for a column vector $\bm{a}$, $||\bm{a}||_2^2=\mbox{Tr}(\bm{a}\bm{a}')$ where $\bm{a}'$ is the transpose of $\bm{a}$.

On the other hand, given $\mathcal{F}_{t-1}$, $\bm{x}(t-1) \sim \mathcal{N}(\hat{\bm{x}}(t-1), \bm{R}(t-1))$ where $(\hat{\bm{x}}(t-1), \bm{R}(t-1))$ can be computed by a standard Kalman filter. Now, 

\footnotesize
\begin{eqnarray}\label{eqn:z-evolution}
 \tilde{\bm{z}}_k(t) 
 &=& \bm{T}_k \bm{z}_k(t) + \bm{b}_k(t) \nonumber\\
 &=& \bm{T}_k \bm{y}_k(t)-\bm{T}_k \bm{H}_k \bm{A} \hat{\bm{x}}^{(k)}(t-1)+ \bm{b}_k(t) \nonumber\\
  &=& \bm{T}_k (\bm{H}_k \bm{x}(t)+\bm{v}_k(t))-\bm{T}_k \bm{H}_k \bm{A} \hat{\bm{x}}^{(k)}(t-1)+ \bm{b}_k(t) \nonumber\\
    &=& \bm{T}_k \bm{H}_k \bm{A} \bm{x}(t-1)+ \bm{T}_k \bm{H}_k \bm{w}(t-1)+\bm{T}_k \bm{v}_k(t) \nonumber\\
    &&-\bm{T}_k \bm{H}_k \bm{A} \hat{\bm{x}}^{(k)}(t-1)+ \bm{b}_k(t)
\end{eqnarray}
\normalsize

which, given $\mathcal{F}_{t-1}$, is distributed as $\mathcal{N}(\bm{T}_k \bm{H}_k \bm{A} \hat{\bm{x}}(t-1)-\bm{T}_k \bm{H}_k \bm{A} \hat{\bm{x}}^{(k)}(t-1)+\bm{M}_k \bm{\theta}^{(k)}(t-1)+\bm{d}_k, \bm{T}_k \bm{H}_k \bm{Q} \bm{H}_k' \bm{T}_k'+ \bm{T}_k \bm{R}_k \bm{T}_k'+ \bm{T}_k \bm{H}_k \bm{A} \bm{R}(t-1) \bm{A}' \bm{H}_k' \bm{T}_k'+\bm{S}_k)$. Hence, $\mathbb{E}(\tilde{\bm{z}}_k(t)' \bm{\Sigma}_k^{-1} \tilde{\bm{z}}_k(t)|\mathcal{F}_{t-1})$ is given by \eqref{eqn:variance-of-z}.

\section{Proof of Lemma~\ref{lemma:convexity}}\label{appendix:proof-of-convexity}
The proof uses the fact that the function $||\sum_{i=1}^n c_i v_i +c||_2^2$ for any arbitrary real known  coefficients $\{c_i\}_{1 \leq i \leq n}$ and $c$ and scalar variables $\{v_i\}_{1 \leq i \leq n}$ is convex in $\{v_i\}_{1 \leq i \leq n}$, since Hessian of this function will be $[c_1, c_2,\cdots,c_n]'[c_1, c_2,\cdots,c_n]$ which is a positive semi-definite matrix. Hence, the first term in the R.H.S. of \eqref{eqn:variance-of-theta} is convex in the arguments. Just as another example, let us consider another term $\mbox{Tr}(\bm{\Sigma}_k^{-\frac{1}{2}} \bm{S}_k \bm{\Sigma}_k^{-\frac{1}{2}})$ from   \eqref{eqn:variance-of-z}; this can be rewritten as $\mbox{Tr}(\bm{\Sigma}_k^{-\frac{1}{2}} \bm{U}_k \bm{U}_k' \bm{\Sigma}_k^{-\frac{1}{2}})=||\bm{\Sigma}_k^{-\frac{1}{2}} \bm{U}_k||_F^2$ which is convex in $\bm{U}_k$ since $\bm{\Sigma}_k^{-\frac{1}{2}} \bm{U}_k$ is a linear function of $\bm{U}_k$. Convexity of other terms can be proven in a similar way.

\section{Proof of Lemma~\ref{lemma:proof-of-stability-of-theta}}
\label{appendix:proof-of-stability-of-theta}
Let us consider the evolution of $\bm{\theta}^{(k)}(t)$ in \eqref{eqn:theta-evolution}, and let $\bm{\theta}(t)$ be the vertical concatenation of the column vectors $\{ \bm{\theta}^{(k)}(t) \}_{1 \leq k \leq N}$. Hence, the evolution of $\bm{\theta}(t)$ is given by: $\bm{\theta}(t)=\bm{M} \bm{\theta}(t-1)+\bm{\zeta}_t$ where $\bm{\zeta}_t$ is a stable  Gaussian proces since $\bm{\phi}(t)$ is a stable process. Hence, $\{\bm{\theta}(t)\}_{t \geq 0}$ is a stable process if the spectral radius of $\bm{M}$ is less than $1$.

\section{Proof of Theorem~\ref{theorem:convergence-of-lambda-OLAADE-KKT-1}} \label{appendix:proof-of-convergence-of-lambda-OLAADE-KKT-1}
Note that, the $\{\bm{M}_k(t), \bm{d}_k(t)\}_{1 \leq k \leq N}$ update and hence the evolution of $g_t(\cdot)$ runs in a faster timescale, while the $\lambda(t)$ update runs in a slower timescale. Also $g_{t,\lambda}(\cdot)$ and $g_{\lambda}^*(\cdot)$ are continuously differentiable in $\lambda$ over a compact interval $[0,A_0]$, and hence are Lipschitz continuous. Clearly, by an argument similar to \cite[Chapter~$6$, Lemma~$1$]{borkar08stochastic-approximation-book}, we claim that $\lim_{t \rightarrow \infty}||g_t(\cdot)-g_{t,\lambda(t)}(\cdot)||_{TV}=0$ almost surely. This proves convergence in faster timescale. 

Now we will prove convergence in the slower timescale. 
Note that, using the fact that $\lambda(t) \in [0,A_0]$ for all $t \geq 0$, and using Assumption~\ref{assumption:M-has-max-eigenvalue-less-than-1} and Lemma~\ref{lemma:stability-of-innovation}, we can easily say that $\{ \tilde{\tilde{\bm{z}}}_k(t) \}_{1 \leq k \leq N}$ is stable under $\mu_{\lambda(t-1), \{\bm{M}_k(t-1), \bm{d}_k(t-1)\}_{1 \leq k \leq N} }$. Also,  note that $\{\bm{x(t)}, \hat{\bm{x}}^{(k)}(t), \bm{y}_k(t), \tilde{\tilde{z}}_k(t),   \bm{M}_k(t), \bm{d}_k(t)\}_{1 \leq k \leq N} \}_{t \geq 0}$ is a stable Markov chain under any $\mu_{\lambda(t-1), \{\bm{M}_k(t-1), \bm{d}_k(t-1)\}_{1 \leq k \leq N} }$ with $\lambda(t-1) \in [0,A_0]$. Hence, the $\lambda(t)$ iteration can be written as:

\footnotesize
\begin{eqnarray*}
    \lambda(t+1)&=&[\lambda(t-1)+b(t) (  \sum_{k=1}^N \mathbb{E}_{ \mu_{\lambda(t-1), \{\bm{M}_k(t-1), \bm{d}_k(t-1)\}_{1 \leq k \leq N} } } \\
    && \bigg( \tilde{\tilde{\bm{z}}}_k(t) \Sigma_k^{-1}  \tilde{\tilde{\bm{z}}}_k(t) \bigg) -\frac{\alpha J}{\eta} + \zeta_1(t))]_0^{A_0}
\end{eqnarray*}
\normalsize

where $\zeta_1(t) \doteq \sum_{k=1}^N  \tilde{\tilde{\bm{z}}}_k(t) \Sigma_k^{-1}  \tilde{\tilde{\bm{z}}}_k(t)  -\sum_{k=1}^N \mathbb{E}_{\mu_{\lambda(t-1), \{\bm{M}_k(t-1), \bm{d}_k(t-1)\}_{1 \leq k \leq N} }} \bigg( \tilde{\tilde{\bm{z}}}_k(t) \Sigma_k^{-1}  \tilde{\tilde{\bm{z}}}_k(t) \bigg)$ is a zero-mean Martingale difference noise. Now, 

\footnotesize
\begin{eqnarray*}
&& \sum_{k=1}^N \mathbb{E}_{\mu_{\lambda(t-1), \{\bm{M}_k(t-1), \bm{d}_k(t-1)\}_{1 \leq k \leq N} }} \bigg( \tilde{\tilde{\bm{z}}}_k(t) \Sigma_k^{-1}  \tilde{\tilde{\bm{z}}}_k(t) \bigg) \\
&=&\lim_{\tau \rightarrow \infty}\sum_{k=1}^N \mathbb{E}_{\mu_{\lambda(t-1), \{\bm{M}_k(t-1), \bm{d}_k(t-1)\}_{1 \leq k \leq N} }} \\
&& \bigg( \tilde{\tilde{\bm{z}}}_k(\tau) \Sigma_k^{-1}  \tilde{\tilde{\bm{z}}}_k(\tau) \bigg)+ o(1)\\
&=& \sum_{k=1}^N \mathbb{E}_{\mu_{\lambda(t-1), \{\bm{M}_k(t-1), \bm{d}_k(t-1)\}_{1 \leq k \leq N} }}  \bigg( \tilde{\tilde{\bm{z}}}_k(\infty) \Sigma_k^{-1}  \tilde{\tilde{\bm{z}}}_k(\infty) \bigg)+ o(1) \\
&=& \sum_{k=1}^N \mathbb{E}_{ \{\bm{M}_k, \bm{d}_k\}_{1 \leq k \leq N}\sim g_{t,\lambda(t-1)}(\cdot) } \mathbb{E}_{ \mu_{\lambda(t-1), \{\bm{M}_k, \bm{d}_k\}_{1 \leq k \leq N}}} \\
&& \bigg( \tilde{\tilde{\bm{z}}}_k(\infty) \Sigma_k^{-1}  \tilde{\tilde{\bm{z}}}_k(\infty) \bigg)+ o(1)+\zeta_2(t) \\
&=& \sum_{k=1}^N \mathbb{E}_{ \{\bm{M}_k, \bm{d}_k\}_{1 \leq k \leq N}\sim g_{\lambda(t-1)}^*(\cdot) } \mathbb{E}_{\lambda(t-1), \mu_{\{\bm{M}_k, \bm{d}_k\}_{1 \leq k \leq N}}} \\
&& \bigg( \tilde{\tilde{\bm{z}}}_k(\infty) \Sigma_k^{-1}  \tilde{\tilde{\bm{z}}}_k(\infty) \bigg)+o(1)+o(1)+\zeta_2(t) \\
\end{eqnarray*}
\normalsize

where the first equality follows from the stability of the above Markov chain, and the second equality follows from the dominated convergence theorem. The third equality uses the fact that $X=\mathbb{E}(X)+X-\mathbb{E}(X)$, with $\zeta_2(t)$ being a Martingale difference noise. The fourth equality follows from the fact that $\lim_{t \rightarrow \infty}||g_{t,\lambda}(\cdot)-g_{\lambda}^*(\cdot)||_{TV}=0$ and the dominated convergence theorem.

Hence, the $\lambda(t)$ iteration can be rewritten as:

\footnotesize
\begin{eqnarray*}
    \lambda(t+1)&=&\bigg[\lambda(t-1)+b(t) ( \sum_{k=1}^N \mathbb{E}_{ \{\bm{M}_k, \bm{d}_k\}_{1 \leq k \leq N}\sim g_{\lambda(t-1)}^*(\cdot) } \\
    && \mathbb{E}_{\mu_{\lambda(t-1), \{\bm{M}_k, \bm{d}_k\}_{1 \leq k \leq N}}} \\
&& \bigg( \tilde{\tilde{\bm{z}}}_k(\infty) \Sigma_k^{-1}  \tilde{\tilde{\bm{z}}}_k(\infty) \bigg)+ \zeta_1(t) +\zeta_2(t)+o(1) \bigg]_0^{A_0}
\end{eqnarray*}
\normalsize

Now, since $g_{\lambda}^*(\cdot)$ is continuous in $\lambda$, we can say that  $ \mathbb{E}_{ \{\bm{M}_k, \bm{d}_k\}_{1 \leq k \leq N}\sim g_{\lambda(t-1)}^*(\cdot) } 
     \mathbb{E}_{\mu_{\lambda(t-1), \{\bm{M}_k, \bm{d}_k\}_{1 \leq k \leq N}}}  ( \tilde{\tilde{\bm{z}}}_k(\infty) \Sigma_k^{-1}  \tilde{\tilde{\bm{z}}}_k(\infty) )$ is continuously differentiable in   $\lambda(t-1) \in [0, A_0]$ and hence Lipschitz continuous.  Also, the offset $o(1)$ goes to $0$ as $t \rightarrow \infty$. Hence, by the theory of basic stochastic approximation  \cite[Chapter~$2$]{borkar08stochastic-approximation-book}, two-timescale  stochastic approximation  \cite[Chapter~$6$]{borkar08stochastic-approximation-book} and projected stochastic approximation  \cite[Chapter~$5$]{borkar08stochastic-approximation-book}, we can say that $\lambda(t) \rightarrow \Lambda$ almost surely.

{\small
\bibliographystyle{unsrt}
\bibliography{arpantechreport}
}

%
%

%
%

\end{document}